\documentclass[11pt]{article}
\usepackage{fullpage}
\usepackage{booktabs}

\usepackage{amsmath}
\usepackage{amssymb}
\numberwithin{equation}{section}

\newcommand\numberthis{\addtocounter{equation}{1}\tag{\theequation}}

\usepackage[section]{algorithm}
\usepackage{algpseudocode}
\usepackage{float}
\usepackage{wrapfig}

\usepackage{mdframed}

\def\showauthornotes{1}
\def\showdraftbox{0}
\def\showcolorlinks{1}

\usepackage{etex}

\usepackage[l2tabu, orthodox]{nag}

\usepackage[utf8]{inputenc}

\usepackage{xspace,enumerate}

\usepackage[dvipsnames]{xcolor}

\usepackage[T1]{fontenc}
\usepackage[full]{textcomp}

\usepackage[american]{babel}

\usepackage{mathtools}

\usepackage{bm}

\usepackage{amsthm}

\newtheorem{theorem}{Theorem}[section]
\newtheorem*{theorem*}{Theorem}

\newtheorem{proposition}[theorem]{Proposition}
\newtheorem*{proposition*}{Proposition}
\newtheorem{lemma}[theorem]{Lemma}
\newtheorem*{lemma*}{Lemma}
\newtheorem{corollary}[theorem]{Corollary}
\newtheorem*{conjecture*}{Conjecture}

\newtheorem*{fact*}{Fact}

\newtheorem*{hypothesis*}{Hypothesis}

\theoremstyle{definition}
\newtheorem{definition}[theorem]{Definition}

\theoremstyle{remark}
\newtheorem{claim}[theorem]{Claim}
\newtheorem*{claim*}{Claim}
\newtheorem{remark}[theorem]{Remark}
\newtheorem*{remark*}{Remark}

\newtheorem*{observation*}{Observation}

\ifnum\showcolorlinks=1
\usepackage[
pagebackref,
letterpaper=true,
colorlinks=true,
urlcolor=blue,
linkcolor=blue,
citecolor=OliveGreen,
]{hyperref}
\fi

\ifnum\showcolorlinks=0
\usepackage[
pagebackref,
letterpaper=true,
colorlinks=false,
pdfborder={0 0 0}
]{hyperref}
\fi

\usepackage{prettyref}

\newcommand{\savehyperref}[2]{\texorpdfstring{\hyperref[#1]{#2}}{#2}}

\newrefformat{eq}{\savehyperref{#1}{\textup{(\ref*{#1})}}}
\newrefformat{lem}{\savehyperref{#1}{Lemma~\ref*{#1}}}
\newrefformat{def}{\savehyperref{#1}{Definition~\ref*{#1}}}
\newrefformat{thm}{\savehyperref{#1}{Theorem~\ref*{#1}}}
\newrefformat{cor}{\savehyperref{#1}{Corollary~\ref*{#1}}}
\newrefformat{cha}{\savehyperref{#1}{Chapter~\ref*{#1}}}
\newrefformat{sec}{\savehyperref{#1}{Section~\ref*{#1}}}
\newrefformat{app}{\savehyperref{#1}{Appendix~\ref*{#1}}}
\newrefformat{tab}{\savehyperref{#1}{Table~\ref*{#1}}}
\newrefformat{fig}{\savehyperref{#1}{Figure~\ref*{#1}}}
\newrefformat{hyp}{\savehyperref{#1}{Hypothesis~\ref*{#1}}}
\newrefformat{alg}{\savehyperref{#1}{Algorithm~\ref*{#1}}}
\newrefformat{rem}{\savehyperref{#1}{Remark~\ref*{#1}}}
\newrefformat{item}{\savehyperref{#1}{Item~\ref*{#1}}}
\newrefformat{step}{\savehyperref{#1}{step~\ref*{#1}}}
\newrefformat{conj}{\savehyperref{#1}{Conjecture~\ref*{#1}}}
\newrefformat{fact}{\savehyperref{#1}{Fact~\ref*{#1}}}
\newrefformat{prop}{\savehyperref{#1}{Proposition~\ref*{#1}}}
\newrefformat{prob}{\savehyperref{#1}{Problem~\ref*{#1}}}
\newrefformat{claim}{\savehyperref{#1}{Claim~\ref*{#1}}}
\newrefformat{relax}{\savehyperref{#1}{Relaxation~\ref*{#1}}}
\newrefformat{red}{\savehyperref{#1}{Reduction~\ref*{#1}}}
\newrefformat{part}{\savehyperref{#1}{Part~\ref*{#1}}}

\ifnum\showauthornotes=1
\newcommand{\Authornote}[2]{{\sffamily\small\color{red}{[#1: #2]}}}
\newcommand{\Authornotecolored}[3]{{\sffamily\small\color{#1}{[#2: #3]}}}
\newcommand{\Authorcomment}[2]{{\sffamily\small\color{gray}{[#1: #2]}}}
\newcommand{\Authorstartcomment}[1]{\sffamily\small\color{gray}[#1: }

\newcommand{\Authorfnote}[2]{\footnote{\color{red}{#1: #2}}}
\newcommand{\Authorfixme}[1]{\Authornote{#1}{\textbf{??}}}
\newcommand{\Authormarginmark}[1]{\marginpar{\textcolor{red}{\fbox{\Large #1:!}}}}
\else
\newcommand{\Authornote}[2]{}
\newcommand{\Authornotecolored}[3]{}
\newcommand{\Authorcomment}[2]{}
\newcommand{\Authorstartcomment}[1]{}

\newcommand{\Authorfnote}[2]{}
\newcommand{\Authorfixme}[1]{}
\newcommand{\Authormarginmark}[1]{}
\fi

\usepackage{boxedminipage}

\newcommand{\problemmacro}[1]{\texorpdfstring{\textsc{#1}}{#1}\xspace}

\ifnum\showdraftbox=1
\newcommand{\draftbox}{\begin{center}
  \fbox{%
    \begin{minipage}{2in}%
      \begin{center}%
          \Large\textsc{Working Draft}\\%
        Please do not distribute%
      \end{center}%
    \end{minipage}%
  }%
\end{center}
\vspace{0.2cm}}
\else
\newcommand{\draftbox}{}
\fi

\usepackage{enumitem}

\newcommand{\clp}[1]{\ensuremath{CLP(#1)}}
\newcommand{\conf}[1]{\ensuremath{\mathcal{C}(#1)}}

\title{Lazy Local Search Meets Machine Scheduling}

\author{
Chidambaram Annamalai\thanks{Department of Computer Science,
    ETH Zurich. Email:
\href{mailto:cannamalai@inf.ethz.ch}{cannamalai@inf.ethz.ch}.}}
\date{\today}

\begin{document}

\maketitle
\draftbox
\thispagestyle{empty}

\begin{abstract}
  We study the \emph{restricted case} of \problemmacro{Scheduling on
    Unrelated Parallel Machines}. In this problem, we are given a set
  of jobs $J$ with processing times $p_j$ and each job may be
  scheduled only on some subset of machines $S_j \subseteq M$. The
  goal is to find an assignment of jobs to machines to minimize the
  time by which all jobs can be processed. In a seminal paper,
  Lenstra, Shmoys, and Tardos~\cite{lenstra1987approximation} designed
  an elegant $2$-approximation for the problem in 1987. The question
  of whether approximation algorithms with \emph{better} guarantees
  exist for this classic scheduling problem has since remained a
  source of mystery.

  \quad In recent years, with the improvement of our understanding of
  Configuration LPs, it now appears an attainable goal to design such
  an algorithm. Our main contribution is to make progress towards this
  goal. When the processing times of jobs are either $1$ or
  $\epsilon \in (0,1)$, we design an approximation algorithm whose
  guarantee tends to $1 + \sqrt{3}/2 \approx 1.8660254,$ for the
  interesting cases when $\epsilon \to 0$. This improves on the
  $2-\epsilon_0$ guarantee recently obtained by Chakrabarty, Khanna,
  and Li~\cite{chakrabarty20151} for some constant $\epsilon_0 > 0$.


\end{abstract}



\newpage
\section{Introduction}
We study a special case of the problem of \problemmacro{Scheduling on
  Unrelated Parallel Machines}. An instance $\mathcal{I}$ of this
problem consists of machines $M$, jobs $J$, and a collection of
positive processing times $\{p_{ij}\}_{i \in M, j \in J}$ for every
machine--job pair. The goal is to assign all the jobs to the available
machines such that the \emph{makespan} of the resulting schedule is as
small as possible. The makespan of a schedule $\sigma : J \mapsto M$
is defined as \[\max_{i \in M} \sum_{j \in \sigma^{-1}(i)} p_{ij}.\]
It is one of the major open questions~\cite{williamson2011design} in
the field of approximation algorithms to better understand the
approximability of this problem. Curiously, the best known hardness
result \emph{and} approximation algorithm for the problem today were
simultaneously established in the 1987 paper of Lenstra, Shmoys and
Tardos~\cite{lenstra1987approximation}.

Given the challenging nature of the problem, the road to a better
understanding of its approximability has focused on two special cases
that each isolate two difficult aspects associated with it---on the
one hand a given job $j \in J$ may be assigned with a finite
processing time to an unbounded number of machines, and on the other
hand, its processing time may vary considerably across those
machines. In the \emph{graph balancing} case, only instances where
every job may be assigned to at most two machines with finite
processing time are considered. In the \emph{restricted} case, often
referred to as the \problemmacro{restricted assignment makespan
  minimization} problem, the input processing times obey the condition
$p_{ij}\in \{\infty, p_j\}$ for each $i \in M$ and $j \in J$, where
$p_j$ is a \emph{machine independent} processing time for job
$j$. This has the natural interpretation that each job has a fixed
processing time but may only be scheduled on some subset of the
machines. The latter special case is the focus of this work.

The elegant $2$-approximation of Lenstra et
al.~\cite{lenstra1990approximation} works by rounding extreme point
solutions to a linear program called the Assignment LP. As this linear
program has a matching integrality gap, one of the natural directions
was to develop a stronger convex relaxation for the \emph{restricted}
case. An important step forward was made by Bansal and
Sviridenko~\cite{bansal2006santa} who, among other things, introduced
the Configuration LP for the problem, which has exponentially many
decision variables. At the time, however, it was not clear if this new
linear program was indeed stronger than the Assignment LP in the sense
of a worst case integrality gap. A breakthrough in this direction was
achieved by Svensson~\cite{svensson2012santa} who proved that the
integrality of the Configuration LP is no worse than
$33/17 \approx 1.94$, and, therefore, strictly better than the
Assignment LP. Tantalizingly, however, the proof of his result did not
lead to an approximation algorithm with the same (or even similar)
guarantee. This is a fairly strange situation for a problem, as we
usually expect integrality gap upper bounds to accompany an
approximation algorithm; indeed, it is often established as a
\emph{consequence} of the latter.

The difficulties in turning the non-constructive aspects of Svensson's
result into an efficient algorithm mirror the situation faced in the
\problemmacro{restricted max-min fair allocation} problem. In the
latter problem, following a long line of
work~\cite{bansal2006santa,feige2008allocations,haeupler2011new,asadpour2012santa,polacek2012quasi,DBLP:conf/soda/AnnamalaiKS15},
a non-constructive integrality gap upper bound on the Configuration LP
by Asadpour, Feige and Saberi~\cite{asadpour2012santa} was turned into
an approximation
algorithm~\cite{DBLP:conf/soda/AnnamalaiKS15}. Although one might
reasonably hope for a similar resolution in the \emph{restricted}
case, it has proved to be elusive thus far. Indeed, there have been
works aimed at obtaining better-than-2 approximation algorithms for
special cases, and the gap between the $33/17 \approx 1.94$ (recently
improved to $11/6 \approx 1.83$~\cite{jansen2016configuration})
integrality gap upper bound and the $2$ approximation algorithm of
Lenstra et al.~\cite{lenstra1990approximation} persists. Ebenlendr,
K{\v{r}}{\'c}al and Sgall~\cite{ebenlendr2008graph} studied the case
when jobs may be assigned to at most two machines and gave a $1.75$
approximation in this case. Recently, Chakrabarty, Khanna, and
Li~\cite{chakrabarty20151} designed a $2-\epsilon_0$ approximation
algorithm, for some constant $\epsilon_0 > 0$, for the so-called
$(1,\epsilon)$-case where processing times of jobs are drawn from a
set of size two.

The special significance of $(1,\epsilon)$-case is that it already
proves to be hard from the perspective of the Configuration LP---the
best known $1.5$ factor integrality gap instances are of this
type---and seems to adequately capture the difficulty of the
\emph{restricted} case in general. It is also interesting in its own
right, and in a sense it is the simplest case of the problem that is
not yet fully understood. Indeed, the case when processing times of
all the jobs are equal can be solved optimally in polynomial
time: this amounts to finding a maximum flow in an appropriate flow
network with jobs as sources, machines as sinks, and setting the sink
capacities to be uniformly equal to some guess on the optimum
makespan.

\paragraph{Our Results.} After normalizing the job sizes, we assume
without loss of generality that the jobs are of size $1$ or $\epsilon$
for some $0 < \epsilon < 1$. Our main result is a new purely flow
based local search algorithm for the $(1,\epsilon)$-case.

\begin{theorem}\label{thm:one-epsilon}
  Let $0 < \epsilon < 1$. For an arbitrary but fixed $\zeta > 0$, the
  $(1,\epsilon)$-case of \problemmacro{restricted assignment makespan
    minimization} admits a polynomial time $1 + R(\epsilon, \zeta)$
  approximation algorithm where
  \[R(\epsilon, \zeta) \stackrel{\Delta}{=}
    \frac12 \left(\sqrt{3-2 \epsilon} + \epsilon\right) + \zeta.\]
\end{theorem}

From the point of view of better-than-$2$ approximation algorithms,
the hard case is when $\epsilon \to 0$. For this range of values, the
approximation ratio guaranteed by Theorem~\ref{thm:one-epsilon} tends
to $1 + \sqrt{3}/2 \approx 1.87$. By balancing against a simple
algorithm based on bipartite matching we also derive an approximation
guarantee independent of the size of small jobs.

\begin{theorem}\label{thm:one-epsilonc}
  Let $0 < \epsilon < 1$. For an arbitrary but fixed $\zeta > 0$, the
  $(1,\epsilon)$-case of \problemmacro{restricted assignment makespan
    minimization} admits a polynomial time $17/9 + \zeta$
  approximation algorithm. Note that $17/9 \approx 1.89$.
\end{theorem}

\paragraph{Our Techniques.} We now give a very high level overview of
the ideas behind the proof of Theorem~\ref{thm:one-epsilon} assuming
that the optimum makespan is $1$ for simplicity.

Our local search algorithm continually increases the number of jobs
scheduled by an assignment
$\sigma : M \mapsto J \cup \{\text{TBD}\}$\footnote{TBD for ``to be
  decided''} that satisfies the $1+R$ makespan bound i.e.,
$\sum_{j \in \sigma^{-1}(i)} p_{ij} \leq 1 + R$ for each $i \in
M$. The algorithm takes a job $j_0$ such that
$\sigma(j_0) = \text{TBD}$ and attempts to assign it to one of the
machines $M$ while respecting the makespan bound of $1+R$. In general,
it may be required to modify the assignment $\sigma$, which we call a
partial schedule, before $j_0$ can be successfully assigned along with
the rest of the jobs in $\sigma^{-1}(M)$. The algorithm identifies a
set of machines $M_0$ such that sufficiently reducing the load on any
of the machines $i \in M_0$ suffices to assign $j_0$
successfully. Once again, to reduce the load on one of the machines in
$M_0$ a new set of machines $M_1$ that is disjoint from $M_0$ is
identified. In this way, in general there will be a sequence of
disjoint machine sets $M_0, M_1, \dots, M_\ell$ such that reducing the
load on some $M_i$ allows the load to be reduced on some machine in
$M_0 \cup \ldots \cup M_{i-1}$. At some point, it may turn out that a
lot of machines in $M_\ell$ have very little load, call them
\emph{free} machines, and therefore, it is possible for many of the
jobs currently scheduled on some machine in
$M_0 \cup \ldots \cup M_{\ell-1}$ to be relocated to free machines in
$M_\ell$, which represents progress towards our goal of eventually
scheduling $j_0$.

The first property we require is \emph{large progress} which ensures
that many free machines in $M_\ell$ implies that proportionally many
jobs from $\sigma^{-1}(M_0 \cup \ldots \cup M_{\ell-1})$ can be
relocated. Further, such relocations should be performed in a way as
to maintain the same property for the machine sets of smaller indices.

Large progress by itself would not guarantee that the algorithm
terminates quickly if it does not happen often enough. To ensure that
we find many free machines \emph{frequently}, a second property
ensures $|M_i| \geq \mu |M_0 \cup \ldots \cup M_{i-1}|$ so that large
progress happens at least once every $O_\mu(\log |M|)$ sets that are
encountered.

These two properties--and maintaining them as the partial schedule is
continually modified by the algorithm--roughly correspond to the core
technical difficulties of our approach. The proof of the first
property (see Section~\ref{sec:maintainingmulti}) makes fairly extensive
use of properties of maximum flows, while the second property (see
Section~\ref{sec:proportionally}) is based on new ideas for constructing
dual unboundedness certificates for the Configuration LP. It is
interesting to note that the construction of such certificates in the
analysis, to prove the required property about the local search
algorithm, effectively amounts to a second algorithm that is merely
used to determine an assignment of values to the dual variables in the
proof.

Our final algorithm is a highly structured local search that performs
job relocations \emph{only} through flow computations in two kinds of
flow networks. In essence, our approach can be seen as effectively
reducing the scheduling problem we started with to polynomially many
maximum flow computations in a structured fashion. The sets alluded to
earlier correspond to certain ``reachability graphs'' associated with
$\epsilon$ jobs whereas the number of such sets at any time is bounded
by a function of the number of $1$ jobs in the instance. We also
stress that identifying such sets, which allow our analyses to prove
these properties, requires a careful design of the algorithm in the
first place, and is not clear at the outset that this can be achieved.






\paragraph{Organization.} The rest of the paper is organized as
follows. In Section~\ref{sec:prelims}, we briefly go over some
notation and state the Configuration LP for our problem. In
Section~\ref{sec:flownetworks} we explain some of the basic concepts
such as flow networks and subroutines for job relocations used in the
final local search algorithm. The algorithm and its running time
analysis are presented in Sections~\ref{sec:lls}
and~\ref{sec:running-time} respectively, followed by the proof of the
main theorem in the paper in Section~\ref{sec:mainthm}. Proofs of some
statements missing in the body of the paper appear in
Appendix~\ref{sec:appendix}.



\section{Preliminaries}\label{sec:prelims}

\subsection{Notation and Conventions}\label{sec:notation}
Let $\mathcal{I}$ be the given instance of \problemmacro{restricted
  assignment makespan minimization}. By scaling all the processing
times in the instance, we assume without loss of generality that
$p_{\max} \stackrel{\Delta}{=} \max \{p_j \; | \; j \in J\} =
1$. We use \texttt{OPT} to denote the optimal makespan for
$\mathcal{I}$.

For a subset of jobs $S \subseteq J$, we use the
notation $p(S)$ and $p_i(S)$ to refer to $\sum_{j \in S} p_j$ and
$\sum_{j \in S} p_{ij}$ respectively. We define
$\Gamma: J \mapsto 2^M$ to be a function that maps each job $j \in J$
to the set of all machines that it can be assigned to with a finite
processing time. The input processing times
$\{p_{ij}\}_{i \in M, j \in J}$ and $\{p_j\}_{j \in J}$ satisfy
\[p_{ij} =
  \left. \begin{cases}
      p_j, &\text{if} \; i \in \Gamma(j),\\
      \infty, &\text{else}.
    \end{cases} \right.  \quad \forall i \in M\; \forall j \in J.
\]
For a collection of indexed sets $\{S_0,\dots, S_\ell\}$ and 
$0 \leq i \leq \ell$ we use the notation $S_{\leq i}$ to refer to the
set $\bigcup_{j=0}^i S_j$.

\subsection{The Configuration Linear Program}\label{subsec:clp}

The Configuration LP is a feasibility linear program
parametrized by a guess $\tau$ on the value of the optimal makespan for
$\mathcal{I}$, and is simply denoted by $\clp{\tau}$. A
\emph{configuration} for machine $i$ is a set of jobs with a total
processing time at most $\tau$ on machine $i$. The collection of all
such configurations for $i$ is denoted as $\conf{i,
  \tau}$. $\clp{\tau}$ ensures that each machine receives at most one
configuration fractionally, while ensuring that every job is assigned,
also in the fractional sense. The constraints of $\clp{\tau}$ are
described in \eqref{eq:clp}.
\begin{align}
  \begin{split}
    \sum_{C \in \conf{i, \tau}} x_{i C} &\leq 1, \quad \forall i \in M,\\
    \sum_{i \in M} \sum_{C \in \conf{i, \tau} \; : \; j \in C}
    x_{i C}
    &\geq 1, \quad \forall j \in J,\\
    x &\geq 0.
  \end{split}\label{eq:clp}
\end{align}

We can also write the dual of $\clp{\tau}$ as follows.
\begin{align}
  \begin{split}
    \max \; &  \sum_{j \in J} z_j - \sum_{i \in M} y_i\\
    y_i & \geq \sum_{j \in C} z_j, \quad \forall i \in M \; \text{and} \; C \in \conf{i,\tau}, \\
    y, z & \geq 0.
  \end{split}\label{eq:clpdual}
\end{align}

Let $\tau^*$ be the smallest value of $\tau$ for which $\clp{\tau}$ is
feasible. We refer to $\tau^*$ as the value of the Configuration
LP. Observe that $\tau^*$ is a lower bound on \texttt{OPT}. As
$p_{\max} = 1$, $\tau^*$ must be at least $1$.

\section{The $(1, \epsilon)$ Case}\label{sec:one-epsilon}

Let $0 < \epsilon < 1$. In the $(1,\epsilon)$-case of
\problemmacro{restricted assignment makespan minimization}, jobs
$j \in J$ have one of only two possible sizes: $1$ or $\epsilon$. We
partition the jobs accordingly into the sets
$J_b \stackrel{\Delta}{=} \{j \in J \; | \; p_j = 1\}$ and
$J_s \stackrel{\Delta}{=} J \setminus J_b$, which we will refer to as
the sets of \emph{big} and \emph{small} jobs respectively.

For the rest of the section fix some $0 < \epsilon < 1$ and
$\zeta > 0$. As $\epsilon$ and $\zeta$ are fixed constants, we refer
to $R(\epsilon, \zeta)$ as simply $R$. To prove
Theorem~\ref{thm:one-epsilon} we describe an algorithm that terminates
in polynomial time with a schedule of makespan at most $\tau^*+R$ for
the given instance $\mathcal{I}$. This algorithm, described in
Section~\ref{sec:lls}, is a local search algorithm which continually
increases the size of a \emph{partial schedule}.

\begin{definition}[Partial schedule]\label{def:partialschedule}
  A partial schedule is a map $\sigma:J \mapsto M \cup \{\text{TBD}\}$
  such that
  \begin{enumerate}[label=(\alph*)]
  \item\label{partialschedule:1} $\forall \; i \in M$, $p_i(\sigma^{-1}(i)) \leq \tau^* + R$,
  \item\label{partialschedule:2} $J_b \subseteq \sigma^{-1}(M)$, and
  \item\label{partialschedule:3} $\forall \; i \in M$, $|\sigma^{-1}(i) \cap J_b| \leq 1$.
  \end{enumerate}

  The \emph{size} of a partial schedule
  $\sigma$ is $|\sigma^{-1}(M)|$.
\end{definition}

\begin{remark}
  Following the description and analysis of our algorithms, it will
  become clear that solving the Configuration LP on the input instance
  $\mathcal{I}$, in order to determine $\tau^*$, is not necessary. For
  the moment, however, it may be assumed that it is somehow known. We
  remark that $\tau^*$ can be computed in polynomial time upto any
  desired accuracy $\nu > 0$ by using the ellipsoid algorithm with an
  appropriate separation oracle~\cite{bansal2006santa}.
\end{remark}

Of course, a partial schedule $\sigma$ of size $|J|$ is a schedule of
makespan at most $\tau^*+R$ for our instance $\mathcal{I}$. The following
statement ensures that partial schedules exist in the first place.
\begin{lemma}\label{lem:partialschedulesexist}
  Suppose $1 \leq \tau^* < 2$. Then, there is a map from $J_b$ to $M$
  such that i) no two big jobs are mapped to the same machine, and ii)
  every big job $j \in J_b$ is mapped to a machine $i_j \in M$ such
  that $i_j \in \Gamma(j)$. Furthermore, such a map can be computed in
  polynomial time.
\end{lemma}

\subsection{Flow Networks for Job Relocations}\label{sec:flownetworks}

Let $\sigma$ be some partial schedule. We now define several
quantities whose description depends on $\sigma$.

\begin{definition}[Job Assignment Graphs]\label{def:jobassignmentgraphs}
  $G_\sigma=(M \cup J, E)$ is a directed bipartite graph with machines
  and jobs in $\mathcal{I}$ as vertices. The edge set of $G_\sigma$ is
  defined as
  \begin{align*}
  E \stackrel{\Delta}{=} \{ (i, j) \;|\; \exists \; i \in M, j \in J \; : \; \sigma(j)
    = i\} \cup \{ (j, i) \;|\; \exists i \in M, j \in J \; : \;
    \sigma(j) \neq i, \;i \in \Gamma(j)\}.
  \end{align*}
  We define the graph of small job assignments
  $G_\sigma^s\stackrel{\Delta}{=}G_\sigma\setminus J_b$ and the graph
  of big job assignments
  $G_\sigma^b\stackrel{\Delta}{=}G_\sigma\setminus J_s$.
\end{definition}

\begin{definition}[Big and Small Machines]\label{def:bigandsmallmachines}
  Let
  $M_\sigma^b \stackrel{\Delta}{=} \{i \in M \;|\;\sigma^{-1}(i) \cap J_b \neq
  \emptyset\}$ and $M_\sigma^s \stackrel{\Delta}{=} M \setminus M_\sigma^b$, which we
  refer to as \emph{big machines} and \emph{small machines}
  respectively.
\end{definition}

We need to define two flow networks which will facilitate the movement
of the two kinds of jobs we have in our instance $\mathcal{I}$. We
will speak of \emph{the} maximum flow in these flow networks, even
though it may not necessarily be unique. In such cases it is
implicitly assumed that there is fixed rule to obtain a particular
maximum flow given a flow network. We also assume that flow
decompositions of flows in such networks contain only source to sink
paths (no cycles). First, we define the flow network for big jobs.

\begin{definition}[Flow Network for Big Jobs]\label{def:bigflownetwork}
  For collections of machines $S \subseteq M_\sigma^b$, and
  $T \subseteq M_\sigma^s$, the flow network $H_\sigma^b(S, T)$ is
  defined on the directed graph $G^b_\sigma$ as follows. Each machine
  to job arc has a capacity of $1$ whereas all other arcs have
  infinite capacity. $S$ and $T$ are the sources and sinks
  respectively in this flow network. Sinks have vertex capacities of
  $1$. The value of maximum flow in this flow network is denoted as
  $|H^b_\sigma(S, T)|$.
\end{definition}

The interpretation of the flow network $H^b_\sigma(S, T)$ is that it
allows us to compute the maximum number of vertex disjoint paths in
the graph $G^b_\sigma$ between the sets of vertices $S$ and $T$.

\begin{proposition}\label{prop:vertexdisjointpaths}
  For any $S \subseteq M_\sigma^b$ and $T \subseteq M_\sigma^s$, there
  are $|H^b_\sigma(S, T)|$ vertex disjoint paths in
  $G^b_\sigma$ with sources in $S$ and sinks in $T$.
\end{proposition}

These vertex disjoint paths suggest an update of the partial schedule
$\sigma$ in the natural way. Algorithm~\ref{alg:bigupdate} formalizes
the description of this task and Proposition~\ref{prop:bigupdate}
follows easily.

\begin{algorithm}[H]
  \caption{\texttt{BigUpdate}$(\sigma, X)$: Update $\sigma$ using flow
    paths in $X$.}
  \label{alg:bigupdate}
  \begin{algorithmic}
    \Require $\sigma$ is a partial schedule and $X$ is a flow in
    $H^b_\sigma(M^b_\sigma, M^s_\sigma)$ where \[\forall f \in X, \quad
    p(\sigma^{-1}(f \cap M^s_\sigma)) \leq \tau^*+R-1.\]
    \State $P \gets \text{Vertex disjoint paths corresponding to $X$
      as ensured by Proposition~\ref{prop:vertexdisjointpaths}}$.
    \ForAll{$p = i_0, j_0,\dots, j_{k_p-1}, i_{k_p} \in P$}
    \For{$\ell = 0, \dots, k_p-1$}
    \State $\sigma(j_\ell) \gets i_{\ell+1}$.
    \EndFor
    \EndFor
    \State \Return $\sigma$.
  \end{algorithmic}
\end{algorithm}

\begin{proposition}
  \label{prop:bigupdate}
  For any partial schedule $\sigma$, and flow $X$ in
  $H^b_\sigma(M^b_\sigma, M^s_\sigma)$ such that $\forall f \in X$,
  $p(\sigma^{-1}(f \cap M^s_\sigma)) \leq  \tau^*+R-1$,
  $\Call{BigUpdate}{\sigma, X}$ returns a partial schedule $\sigma'$ such that
  \begin{enumerate}[label=(\alph*)]
  \item\label{bigupdate:1} $\sigma'^{-1}(M) = \sigma^{-1}(M)$, and
  \item\label{bigupdate:2} $\forall f = i_0, j_0, \dots, j_{k-1}, i_k \in X, \; \sigma'^{-1}(i_0) \cap J_b = \emptyset$.
  \end{enumerate}
\end{proposition}

Now we define the flow network for small jobs.

\begin{definition}[Flow Network for Small Jobs]\label{def:smallflownetwork}
  For two \emph{disjoint} collections of machines
  $S \subseteq M_\sigma^s$ and $T \subseteq M$, the flow network
  $H^s_\sigma(S, T)$ is defined on the directed graph $G^s_\sigma$ as
  follows. The arcs going from machines to jobs have capacity
  $\epsilon$ while arcs going from jobs to machines have infinite
  capacity. $S$ and $T$ are the sources and sinks respectively in this
  flow network. The sinks have \emph{vertex capacities} are set as
  follows:
  \begin{align*}
    \forall i \in T, \quad c(i) = \left. \begin{cases}
        1+\tau^*+R-p(\sigma^{-1}(i))-\epsilon, &\text{if} \; i \in M_\sigma^b,\\
        \tau^*+R-p(\sigma^{-1}(i)), &\text{else}.
      \end{cases} \right.
  \end{align*}
  The value of the maximum flow in this network is denoted as
  $|H^s_\sigma(S, T)|$.
\end{definition}

By construction it is clear that the maximum flow in both flow
networks~(\ref{def:bigflownetwork}) and~(\ref{def:smallflownetwork})
is finite. By the max-flow min-cut theorem, infinite capacity arcs
going from the source-side to the sink-side will therefore not cross
any minimum capacity. We will use this fact later in our proof.

Algorithm~\ref{alg:smallupdate} interprets flows in $H^s_\sigma(S,T)$
as a collection of updates for $\sigma$.
Proposition~\ref{prop:smallupdate} is a statement about the partial
schedule $\sigma'$ output by $\Call{SmallUpdate}{\ldots}$, and the
flow $X'$ computed at the end of the \textbf{while} loop in the
procedure. For convenience we let $f^{\text{source}}$ and
$f^{\text{sink}}$ denote the source and sink vertices, respectively,
of a flow path $f$.

\begin{algorithm}[H]
  \caption{\texttt{SmallUpdate}$(\sigma, S, T)$: Update $\sigma$ using
    $H^s_\sigma(S, T)$.}
  \label{alg:smallupdate}
  \begin{algorithmic}
    \Require $\sigma$ is a partial schedule, $S \subseteq M_\sigma^s$,
     $T \subseteq M \setminus S$.


    \State $X_0 \gets \text{Maximum flow in the network
      $H^s_\sigma(S, \{i \in T \; | \; p(\sigma^{-1}(i)) \leq \tau^* +
      R - \epsilon\})$}$.
    \State $X \gets \text{Augment $X_0$ to a maximum flow in the network
      $H^s_\sigma(S, T)$}$.

    \While{$\exists f = i_0, j_0,\dots, j_{k_f-1}, i_{k_f} \in X \; : \; p(\sigma^{-1}(f^{\text{sink}})) \leq \tau^* + R - \epsilon$}
    \For{$\ell = 0, \dots, k_f-1$}
    \State $\sigma(j_\ell) = i_{\ell+1}$.
    \EndFor
    \State $X \gets X \setminus \{f\}$.
    \EndWhile
    \State \Return $\sigma$.
  \end{algorithmic}
\end{algorithm}

\begin{proposition}\label{prop:smallupdate}
  Let $S \subseteq M_\sigma^s, T \subseteq M \setminus S$, and $X$ be the maximum
  flow in $H^s_\sigma(S, T)$. Then,
  $\Call{SmallUpdate}{\sigma, S, T}$ returns a partial schedule
  $\sigma'$ and computes a maximum flow $X'$ in $H^s_{\sigma'}(S,T)$ at the end
  of the \textbf{while} loop in Algorithm~\ref{alg:smallupdate} such that
  \begin{enumerate}[label=(\alph*)]
  \item\label{smallupdate:1} $\sigma'^{-1}(M) = \sigma^{-1}(M)$,
  \item\label{smallupdate:2}
    $\forall i \in S, \; p(\sigma'^{-1}(i)) - \epsilon \cdot|\{f \in
    X'\;|\; f^{\text{source}} = i\}| = p(\sigma^{-1}(i)) - \epsilon\cdot|\{f
    \in X\;|\; f^{\text{source}} = i\}|$, and
  \item\label{smallupdate:3} $\forall f \in X, \; p(\sigma'^{-1}(f \cap T)) > \tau^* + R -\epsilon$.
  \item\label{smallupdate:4} There is no path in the graph $G^s_{\sigma'}$ from $S$ to some machine $i \in T$ such that \[p({\sigma'}^{-1}(i)) \leq \tau^*+R-\epsilon.\]
  \end{enumerate}
\end{proposition}

\begin{proof}
  The first three properties follow directly from the updates
  performed in the \textbf{while} loop of
  Algorithm~\ref{alg:smallupdate}. To see that $X'$ is a maximum flow
  in $H^s_{\sigma'}(S,T)$, first observe that for each update of the
  partial schedule $\sigma$ maintained by the algorithm, along some
  flow path $f \in X$ from $\sigma^{(b)}$ to $\sigma^{(a)}$ in a
  single iteration of the \textbf{while} loop (superscripts for before
  and after), the graph $G^s_{\sigma^{(a)}}$ can be obtained from
  $G^s_{\sigma^{(b)}}$ by simply reversing the directions of arcs of
  the flow path in question. Suppose that $X$ is a maximum flow in
  $H^s_{\sigma^{(b)}}(S,T)$. By the max-flow min-cut theorem, the
  maximum flow is associated with a minimum capacity cut of equal
  value, and the latter observation implies that both the flow value
  in $X \setminus \{f\}$ and the corresponding cut capacity in the new
  network $H^s_{\sigma^{(a)}}(S,T)$ is less, than the flow value in
  $X$ and its corresponding cut capacity in $H^s_{\sigma^{(b)}}(S,T)$,
  by $\epsilon$. This implies that $X \setminus \{f\}$ is a maximum
  flow in $H^s_{\sigma^{(b)}}(S,T)$.

  The final property follows from the particular way in which the
  maximum flow $X$ in $H^s_\sigma(S,T)$ is constructed. Note that
  Algorithm~\ref{alg:smallupdate} first computes a flow $X_0$ that
  maximizes the flow to machines in
  $\{i \in T \; | \; p(\sigma^{-1}(i)) \leq \tau^* + R - \epsilon\}$,
  and then augments $X_0$ to a maximum flow in
  $H^s_\sigma(S,T)$. Suppose to the contrary that there is a path $P$
  in $G^s_{\sigma'}$ from $S$ to a machine $i$ such that
  $p(\sigma'^{-1}(i))\leq \tau^*+R-\epsilon$. We know that one can
  obtain $G^s_{\sigma'}$ from $G^s_{\sigma}$ by reversing the arc
  directions of all the flow paths in $X \setminus X'$ i.e., the paths
  that were used to update the partial schedule in the \textbf{while}
  loop. So each arc in $P$ is either i) present in some path of
  $X \setminus X'$ in the opposite direction, or ii) not present in
  the paths of $X \setminus X'$ in the opposite direction and,
  therefore, also present in $G^s_{\sigma}$. Now consider the
  \emph{residual flow network} of the flow $X \setminus X'$ in
  $H^s_{\sigma}(S,T)$. It is now easy to see that $P$ is an augmenting
  path in this residual flow network because
  $p(\sigma^{-1}(i)) + \epsilon \cdot |\{f \in X \setminus X' \; | \;
  f^{\text{sink}} = i\}| = p({\sigma'}^{-1}(i)) \leq
  \tau^*+R-\epsilon$, and, therefore,
  $c(i) \geq \tau^* + R - p(\sigma^{-1}(i)) \geq \epsilon \cdot |\{f
  \in X \setminus X' \; | \; f^{\text{sink}} = i\}| + \epsilon$. This,
  however, contradicts the maximality of the flow $X_0$ computed in
  the first step of the algorithm.
\end{proof}

We now have the tools necessary to state our local search algorithm.

\subsection{Flow Based Local Search}\label{sec:lls}

In this section we describe the local search algorithm that takes as
input $\tau^*$, a partial schedule $\sigma$, and a small job
$j_0 \in J_s \setminus \sigma^{-1}(M)$. It outputs a partial schedule
$\sigma'$ such that $\sigma'^{-1}(M) = \sigma^{-1}(M) \cup
\{j_0\}$. The algorithm is parameterized by three constants
$0 < \mu_1, \mu_2, \delta \leq 1$. They are defined as:
\begin{align}
  \begin{split}
    \mu_1 &\stackrel{\Delta}{=} \min\{1,\zeta\}/4,\\
    \mu_2 &\stackrel{\Delta}{=} \min\{\delta,\zeta\}/4,\\
    \delta &\stackrel{\Delta}{=} \left(\sqrt{3-2 \epsilon }-1\right)/2.
  \end{split}\label{eq:mudelta}
\end{align}

The algorithm maintains certain sets of machines in \emph{layers}
$L_0,\dots,L_\ell$ throughout its execution, where $\ell$ is some
dynamically updated index variable that always points to the last
layer. A layer $L_i$ is a tuple $(A_i, B_i)$ where
$A_i \subseteq M_\sigma^s$, and $B_i \subseteq M$, except $L_0$ which
is defined to be $(\{j_0\}, B_0)$ for some $B_0 \subseteq M$. In
addition to layers, the algorithm also maintains a collection of
machines $\{I_i\}_{i=0}^\ell$ that will be disjoint from the machines
in $L_{\leq \ell}$.

We will describe the algorithm in a procedural style, and in the
course of the execution the algorithm, sets and other variables will
be modified. Function calls are specified in the pseudocode assuming
call-by-value semantics. Concretely, this means that function calls
have no side effects besides the assignment of the returned values at
the call site. We abuse notation slightly and use $L_i$ to also refer
to $A_i \cup B_i$ (for $i=0$, as $A_0$ is not a set of machines, we
use $L_0$ to just mean $B_0$), so that $L_{\leq \ell}$ refers to a set
of machines. For a subset of machines $N \subseteq M$ we use
$R^s_\sigma(N)$ denote the set of all machines reachable from vertices
$N$ in the graph $G_\sigma^s$. Note that $N \subseteq R^s_\sigma(N)$
always holds.

\begin{algorithm}
  \caption{\texttt{LocalSearch}$(\tau^*, \sigma, j_0)$: Extend the partial
    schedule $\sigma$ to include small job $j_0$.}
  \label{alg:localsearch}
  \begin{algorithmic}[1]
    \Require $\sigma$ is a partial schedule, $j_0 \in J_s \setminus \sigma^{-1}(M)$.
    \State Set $A_0 \gets \{j_0\}$, $B_0 \gets R^s_\sigma(\Gamma(j_0))$.\label{step:init0}
    \Comment Construction of layer $L_0$
    \State Set $\ell \leftarrow 0$ and $I_0 \gets \emptyset$.
    \While{$\nexists \; i \in B_0$ such that $p(\sigma^{-1}(i)) \leq \tau^*
      + R - \epsilon$}\label{step:mainloop}
    \Comment Main loop
    \State\label{step:finishconstruction} $(\sigma, A_{\ell+1}, B_{\ell+1}) \gets
    \Call{BuildLayer}{\sigma, \{L_i\}_{i=0}^\ell, \{I_i\}_{i=0}^\ell}$.
    \Comment Construction of layer $L_{\ell+1}$
    \State\label{step:initI} Set $\ell \leftarrow \ell+1$ and $I_{\ell+1}\gets \emptyset$.
    \While{$\ell \geq 1$ and $|\{i \in A_\ell \; | \;
      p(\sigma^{-1}(i)) \leq \tau^* + R - 1\}| \geq \mu_2|A_\ell|$}\label{step:collapseq}
    \State Set $I \gets \{i \in A_{\ell}\; | \;p(\sigma^{-1}(i)) \leq
    \tau^* + R - 1\}$.\label{step:newcollapse}
    \State $(I'_0, \dots, I'_{\ell}, X) \gets
    \Call{CanonicalDecomposition}{\sigma, \{L_i\}_{i=0}^{\ell}, \{I_i\}_{i=0}^{\ell}, I}$.\label{step:canonicald}
    \State Set $I_i \gets I'_i$ for all $1 \leq i \leq \ell$.\label{step:assigncanon}
    \If{$\exists \; 1 \leq r \leq \ell : |I_r| \geq \mu_1\mu_2
      |B_{r-1} \cap M^b_\sigma|$}\label{step:ifr}
    \State Choose the \emph{smallest} such $r$.\label{step:chooser}
    \State $\sigma \gets \Call{BigUpdate}{\sigma, \{f \in X \; | \; f \cap
    I_r \neq \emptyset\}}$.\label{step:bigupdate}
    \State $\sigma \gets \Call{SmallUpdate}{\sigma, A_{r-1},
      B_{r-1}}$ \textbf{unless} $r = 1$.\label{step:smallupdate2}
    \State $B_{r-1} \gets R^s_\sigma(A_{r-1}) \setminus \left(
    A_{r-1}\cup L_{\leq r-2} \cup I_{\leq r-2} \right)$ \textbf{unless} $r = 1$.\label{step:updateB}
    \State Discard all layers with indices greater than $r-1$.\label{step:discard}
    \State Set $\ell \gets r-1$.\label{step:resetell}
    \EndIf
    \EndWhile
    \EndWhile
    \State Update $\sigma$ using a path from $j_0$ to $i$ in
    $G^s_\sigma$ where $p(\sigma^{-1}(i)) \leq \tau^* + R -\epsilon$.\label{step:termination}
    \State \Return $\sigma$.
    \State
    \Function{CanonicalDecomposition}{$\sigma, \{L_i\}_{i=0}^{\ell}, \{I_i\}_{i=0}^{\ell}, I$}
    \State Let $X$ be the maximum flow in
    $H^b_\sigma(B_0 \cap M_\sigma^b, I_{\leq \ell} \cup I)$.
    \For{$1 \leq i \leq \ell-1$}
    \State Augment $X$ to a maximum flow in $H^b_\sigma(B_{\leq i} \cap
    M_\sigma^b, I_{\leq \ell} \cup I)$.
    \EndFor
    \For{$1 \leq i \leq \ell$}
    \State Set $I'_i$ to be the collection of sinks used by flow paths from $X$ with sources in $B_{i-1}$
    \EndFor
    \State \Return $(\emptyset, I'_1,\dots,I'_{\ell}, X)$.
    \EndFunction
  \end{algorithmic}
\end{algorithm}

\begin{algorithm}
  \caption{\texttt{BuildLayer}$(\sigma, \{L_i\}_{i=0}^\ell, \{I_i\}_{i=0}^\ell)$: Construct and return a new layer.}
  \label{alg:buildlayer}
  \begin{algorithmic}[1]
  \State Let $S \leftarrow \emptyset$.\label{step:newphase}
  \While{$\exists i \in M : \Call{IsAddableQ}{i, \sigma, \{L_i\}_{i=0}^\ell, \{I_i\}_{i=0}^\ell, S}$}\label{step:makeS}
  \State $\sigma \gets \Call{SmallUpdate}{\sigma, S\cup\{i\},
    T\setminus \{i\}}$ where $T \gets M
  \setminus \left( L_{\leq \ell} \cup I_{\leq \ell} \cup S\right)$.\label{step:smallupdate}
  \State $S \gets S \cup \{i\}$.
  \EndWhile
  \State $A_{\ell+1} \gets S$.\label{step:assignS}
  \State
  $B_{\ell+1} \gets R^s_\sigma(A_{\ell+1}) \setminus \left(
    A_{\ell+1}\cup L_{\leq \ell} \cup I_{\leq \ell} \right)$.\label{step:makeB}
  \State \Return $(\sigma, A_{\ell+1}, B_{\ell+1})$.
  \State
  \Function{IsAddableQ}{$i, \sigma, \{L_i\}_{i=0}^\ell,
    \{I_i\}_{i=0}^\ell, S$}
  \Comment Decide if $i$ can be added to $S$
  \If{$i \not \in M_\sigma^s \setminus (L_{\leq \ell} \cup
    I_{\leq \ell} \cup S)$}
  \State \Return False.
  \EndIf
  \State Set $T \gets M \setminus \left(L_{\leq \ell} \cup I_{\leq
      \ell} \cup S\right)$.
  \If{$|H^b_\sigma(B_{\leq \ell} \cap M_\sigma^b, S \cup \{i\})| = |H^b_\sigma(B_{\leq \ell} \cap M_\sigma^b, S)| + 1$}
  \If{$|H^s_\sigma(S \cup \{i\}, T \setminus \{i\})| \; \geq \;
    |H^s_\sigma(S, T)| \; + \; \left(p(\sigma^{-1}(i)) - (\tau^* - 1 + R -
      \delta)\right)$}
  \State \Return True.
  \EndIf
  \EndIf
  \State \Return False.
  \EndFunction

  \end{algorithmic}
\end{algorithm}



The description is now found in Algorithm~\ref{alg:localsearch}. We
refer to the \textbf{while} loop in Step~\ref{step:mainloop} of
Algorithm~\ref{alg:localsearch} as the \emph{main loop} of the
algorithm. Observe that, in Step~\ref{step:finishconstruction},
Algorithm~\ref{alg:buildlayer} is used a subroutine, which constructs
and returns a new layer while potentially modifying the partial
schedule $\sigma$ maintained by Algorithm~\ref{alg:localsearch}.

\subsection{Running Time Analysis}\label{sec:running-time}
The \emph{state of the algorithm} is defined to be the dynamic tuple
$\mathcal{S} \stackrel{\Delta}{=} (\sigma, \ell,\{L_i\}_{i=0}^\ell,
\{I_i\}_{i=0}^\ell)$ which contains the variables and sets that are
maintained by Algorithm~\ref{alg:localsearch}. In the analysis it will
be useful to compare quantities before and after certain operations
performed by the algorithm, and we will consistently use $\mathcal{S}$
and $\mathcal{S}'$ to refer to the state of the algorithm before and
after such an operation. For example, if $\mathcal{S}$ and
$\mathcal{S}'$ denote the states of the algorithm before
Step~\ref{step:finishconstruction} and after Step~\ref{step:initI} in
Algorithm~\ref{alg:localsearch} respectively, then $\ell' = \ell+1$,
$I'_{\ell'} = \emptyset$, etc.

\subsubsection{Basic Invariants of the Algorithm}
By observing the description of Algorithms~\ref{alg:localsearch}
and~\ref{alg:buildlayer} we can conclude certain basic properties
which will come in handy when reasoning about its running time.

\begin{proposition}\label{prop:basicinvs}
  Consider some state $\mathcal{S}$ of the algorithm.
  \begin{enumerate}[label=(\alph*)]
  \item\label{basicinvs:1} The sets in the collection
    $\{A_i\}_{i=1}^\ell \cup \{B_i\}_{i=0}^\ell \cup
    \{I_i\}_{i=0}^\ell$ are pairwise disjoint subsets of $M$.
  \item\label{basicinvs:2} For each $i=1,\dots,\ell$, the sets $A_i$
    have not been modified since the last time $L_i$ was initialized
    in some execution of Step~\ref{step:finishconstruction} of
    Algorithm~\ref{alg:localsearch}. Similarly, for $i=0$, the sets
    $A_0$ and $B_0$ have not been modified since the execution of
    Step~\ref{step:init0}.
  \item\label{basicinvs:3} For each $i \in B_{\leq \ell}$,
    $p(\sigma^{-1}(i)) > \tau^* + R -\epsilon$.
  \item\label{basicinvs:4} For every
    $j \in \sigma^{-1}(L_{\leq \ell}) \cap J_s$,
    $\Gamma(j) \subseteq L_{\leq \ell} \cup I_{\leq \ell}$.
  \end{enumerate}
\end{proposition}

\begin{proof}
  \begin{enumerate}[label=(\alph*)]
  \item For a newly constructed layer $L_{\ell+1}$ in
    Step~\ref{step:finishconstruction} of
    Algorithm~\ref{alg:localsearch}, the sets $A_{\ell+1}$ and
    $B_{\ell+1}$ satisfy the properties by the construction of the set
    $S$ and the setting in Step~\ref{step:makeB}
    Algorithm~\ref{alg:buildlayer}. Further, $I_{\ell+1}$ is initialized
    to the empty set in Step~\ref{step:initI} of Algorithm~\ref{alg:localsearch}. In future updates of $B_{\ell+1}$ (if any) in
    Step~\ref{step:updateB} of Algorithm~\ref{alg:localsearch}, let
    $\mathcal{S}$ and $\mathcal{S}'$ be the states before
    Step~\ref{step:smallupdate2} and after Step~\ref{step:updateB} of
    Algorithm~\ref{alg:localsearch} respectively. From the description
    of Algorithm~\ref{alg:smallupdate}, we see that
    $\sigma^{-1}(A_{r-1}) \supseteq \sigma'^{-1}(A_{r-1}) = \sigma'^{-1}(A'_{r-1})$, which then
    implies that $B_{r-1} \supseteq B'_{r-1}$ from the assignment in
    Step~\ref{step:updateB}.
  \item This follows directly from the description of the algorithm.
  \item When a new layer $L_{\ell+1}$ is constructed during a call to
    $\Call{BuildLayer}{\ldots}$, at the end of the \textbf{while} loop
    in Step~\ref{step:makeS} of Algorithm~\ref{alg:buildlayer}, we
    show in Claim~\ref{claim:inv1}, that a maximum flow $X$ in
    $H^s_{\sigma}(S, T)$ is computed where
    $T = M \setminus (L_{\leq \ell} \cup I_{\leq \ell} \cup S)$. So we
    can apply Proposition~\ref{prop:smallupdate}\ref{smallupdate:4}
    and conclude that after the assignment in Step~\ref{step:makeB},
    there is no machine $i \in B_{\ell+1}$ such that
    $p(\sigma^{-1}(i)) \leq \tau^*+R-\epsilon$. We can argue in
    exactly the same way in Steps~\ref{step:smallupdate2}
    and~\ref{step:updateB} of Algorithm~\ref{alg:localsearch}.
  \item This follows from Step~\ref{step:makeB} of
    Algorithm~\ref{alg:buildlayer} and Step~\ref{step:updateB} of
    Algorithm~\ref{alg:localsearch}.
  \end{enumerate}
\end{proof}

\begin{definition}[Collapsibility of a layer]\label{def:collapsibility}
  Layer $L_0$ is \emph{collapsible} if there is an $i \in B_0$ such
  that $p(\sigma^{-1}(i)) \leq \tau^* + R - \epsilon$.  For $\ell \geq 1$,
  $L_\ell$ is \emph{collapsible} if $A_\ell$ contains at least
  $\mu_2|A_\ell|$ machines $i$ such that $p(\sigma^{-1}(i)) \leq
  \tau^* + R - 1$.
\end{definition}

Note the correspondence between Definition~\ref{def:collapsibility}
and the conditions in Steps~\ref{step:mainloop}
and~\ref{step:collapseq} of Algorithm~\ref{alg:localsearch}.

\begin{lemma}\label{lem:noncollapsibility}
  At the beginning of each iteration of the main loop of
  Algorithm~\ref{alg:localsearch}, none of the layers
  $L_0,\dots, L_\ell$ are collapsible. In particular, for all
  $1 \leq i \leq \ell$,
  \[|\{ i' \in A_i \; | \; p(\sigma^{-1}(i')) \leq \tau^* + R - 1\}| <
    \mu_2|A_i|.\]
\end{lemma}
\begin{proof}
  In the first iteration of the main loop of the algorithm, $\ell=0$,
  and the statement follows from the condition of the main loop of the
  algorithm. Assume the statement to hold at the beginning of some
  iteration of the main loop of the algorithm with the state
  $(\sigma, \ell,\{L_i\}_{i=0}^\ell, \{I_i\}_{i=0}^\ell)$. Observe
  that a new layer $L_{\ell+1}$ is created in
  Step~\ref{step:finishconstruction} and it is tested for
  collapsibility in Step~\ref{step:collapseq}. Suppose that the
  \textbf{while} loop in the latter step executes at least once (if it
  happens infinitely many times, we are done). Let $r$ denote the
  choice made in Step~\ref{step:chooser} in the last execution of the
  step. The layers $L_0,\dots,L_{r-2}$ continue to be non-collapsible
  by induction and Proposition~\ref{prop:basicinvs}\ref{basicinvs:2},
  and layer $L_{r-1}$ is not collapsible because execution exits the
  \textbf{while} loop.
\end{proof}

\begin{lemma}\label{lem:smallI}
  At the beginning of each iteration of the main loop of the
  algorithm, for every $0 \leq i \leq \ell-1$,
  \[|I_{i+1}| < \mu_1\mu_2 |B_i \cap M_\sigma^b|.\]
\end{lemma}
\begin{proof}
  The sets $I_0,\dots, I_\ell$ maintained by the algorithm start out
  initialized to $\emptyset$ in Step~\ref{step:initI} of
  Algorithm~\ref{alg:localsearch} when the corresponding layer is
  created. They are modified only within the \textbf{while} loop of
  Step~\ref{step:collapseq} of Algorithm~\ref{alg:localsearch} through the
  computation of the canonical decomposition and assignment in
  Step~\ref{step:assigncanon}. Within this loop, in
  Step~\ref{step:chooser}, the smallest $1 \leq r \leq \ell$ such that
  $|I_r| \geq \mu_1\mu_2 |B_{r-1} \cap M^b_\sigma|$ is chosen; layers
  with indices greater than $r-1$ are discarded in
  Step~\ref{step:discard}; and $\ell$ is set to $r-1$ at the end of
  the loop in Step~\ref{step:resetell}. Therefore, the claim follows.
\end{proof}

\subsubsection{Maximum Flows and Canonical Decompositions}

We now recall some properties about network flows that will be of use
later on in the proof our main theorem. For basic concepts related to
flows, such as residual flow networks and augmenting paths, we refer
the reader to the textbook by Cormen, Leiserson, Rivest and
Stein~\cite{DBLP:books/daglib/0023376}.

\begin{proposition}\label{prop:flows}
  Let $S \subseteq M^b_\sigma$ and $T \subseteq M^s_\sigma$. Let $S'
  \subseteq M^s_\sigma$ and $T' \subseteq M$ such that $S' \cap T' = \emptyset$.
  \begin{enumerate}[label=(\alph*)]
  \item\label{flows:1} Let $X$ be the maximum flow in $H^b_\sigma(S, T)$ and let
    $C_X$ denote the minimum capacity cut corresponding to $X$ i.e.,
    the set of vertices reachable from $S$ in the residual flow
    network of $X$. Then,
    $|H^b_\sigma(S, T \cup \{i\})| > |H^b_\sigma(S, T)|$ for all
    $i \in C_X \setminus S$.
  \item\label{flows:2} Let $Y$ be the maximum flow in
    $H^s_\sigma(S', T')$ and let $C_Y$ denote the minimum capacity cut
    corresponding to $Y$. For any $i \in M \setminus C_Y$ such that
    $i$ is not used as a sink by a flow path in $Y$, let the
    corresponding maximum flow in $H^s_\sigma(S' \cup \{i\}, T'
    \setminus \{i\})$ be $Y'$ and minimum capacity cut be
    $C_{Y'}$. Then, $C_Y \subset C_{Y'}$ and this inclusion is strict.
  \end{enumerate}
\end{proposition}

We state some consequences of the description of the procedure
$\Call{CanonicalDecomposition}{\ldots}$ in
Algorithm~\ref{alg:localsearch}.

\begin{proposition}\label{prop:canonicald}
  For a given state $\mathcal{S}$ and $I \subseteq M^s_\sigma$, let
  $(I'_0, \dots, I'_{\ell}, X)$ be the tuple returned by
  $\Call{CanonicalDecomposition}{\sigma, \{L_i\}_{i=0}^{\ell},
    \{I_i\}_{i=0}^{\ell}, I}$. Then, $X$ is a maximum flow in the
  network
  \[H^b_\sigma(B_{\leq \ell-1} \cap M^b_\sigma, I_{\leq \ell} \cup
    I),\] such that
  \begin{enumerate}[label=(\alph*)]
  \item\label{canonicald:1} $I'_i$ is the collection of sinks used by
    flow paths from $X$ with sources in $B_{i-1}$ for all
    $i=1,\dots,\ell$, and $I'_0=\emptyset$,
  \item\label{canonicald:2}  $|H^b_\sigma(B_{\leq i} \cap M^b_\sigma, I'_{\leq i+1})| =
    |H^b_\sigma(B_{\leq i} \cap M^b_\sigma, I_{\leq \ell}\cup
    I)|,$ for all $i =0,\dots,\ell-1$, and
  \item\label{canonicald:3}
    $|H^b_\sigma(B_{\leq i} \cap M^b_\sigma, I'_{\leq i+1})| =
    |H^b_\sigma(B_{\leq i} \cap M^b_\sigma, I'_{\leq \ell})|,$ for all
    $i =0,\dots,\ell-1$.
  \end{enumerate}
\end{proposition}

\subsubsection{Relating Set Sizes within Layers}

\begin{lemma}\label{lem:bi0}
  Suppose that $1 \leq \tau^* < 2$. At the beginning of each iteration
  of the main loop of the algorithm, $|B_0 \cap M_\sigma^b| \geq 1$.
\end{lemma}
\begin{proof}
  After the execution of Step~\ref{step:init0} of
  Algorithm~\ref{alg:localsearch}, $|A_0| = 1$ and
  $|B_0 \cap M^b_\sigma| \geq 1$. The latter inequality follows from
  the feasibility of $\clp{\tau^*}$ and the fact that
  $\tau^*+R-\epsilon \geq \tau^*$. We omit its proof here since it
  is similar to Lemma~\ref{lem:partialschedulesexist}. In the main
  loop of the algorithm, consider the first time (if at all) the
  schedule of big jobs on machines in $B_0$ is altered in
  Step~\ref{step:bigupdate}. Then it must be that
  $|I_1| \geq \mu_1\mu_2|B_0 \cap M^b_\sigma| > 0$ from
  Step~\ref{step:ifr}. Using
  Proposition~\ref{prop:canonicald}\ref{canonicald:2} in
  Step~\ref{step:canonicald}, $X$ contains a set of flow paths
  connecting sources in $B_0 \cap M^b_\sigma$ to $I_1$. Then,
  Proposition~\ref{prop:bigupdate}\ref{bigupdate:2} implies that
  $|B'_0 \cap M^b_{\sigma'}| < |B_0 \cap M^b_\sigma|$ after
  Step~\ref{step:bigupdate}. Then, the condition of the main loop of
  the algorithm is no longer satisfied since
  $p(\sigma'^{-1}(i)) \leq \tau^* + R - 1 \leq \tau^* + R - \epsilon$
  for some $i \in B'_0$. The condition of the \textbf{while} loop in
  Step~\ref{step:collapseq} is also not satisfied because $\ell = 0$
  after Step~\ref{step:resetell}. Therefore, the main loop is exited
  in this case.
\end{proof}

\begin{lemma}\label{lem:bi}
  At the beginning of each iteration of the main loop of the
  algorithm, for every $1 \leq i \leq \ell$,
  \[|B_i \cap M_\sigma^b| > \left(\delta(1-\mu_2) - 2\mu_2\right)\cdot |A_i|.\]
\end{lemma}
\begin{proof}
  We first prove a general claim which will then ease the proof the lemma.
  \begin{claim}\label{claim:inv1}
    The \textbf{while} loop in Step~\ref{step:makeS} of
    Algorithm~\ref{alg:buildlayer} that iteratively builds $S$ and
    modifies $\sigma$ satisfies the following invariant, where,
    $T \stackrel{\Delta}{=} M \setminus (L_{\leq \ell} \cup I_{\leq
      \ell} \cup S)$, as defined in Step~\ref{step:smallupdate}, and
    $X$ denotes the maximum flow in the network $H^s_\sigma(S, T)$.
    \begin{center}
      \fbox{\parbox{0.8\textwidth}{
          \begin{align*}
            \epsilon |X| \geq \sum_{i' \in S}  &\left(p(\sigma^{-1}(i')) - (\tau^*- 1 + R-\delta)\right),\\
            \forall f \in X,\quad & f^{\text{sink}} \in M^b_\sigma.
          \end{align*}}}
    \end{center}

  \end{claim}
  \begin{proof}
    Before the first iteration of the \textbf{while} loop,
    $S = \emptyset$ and the statement is vacuously true. Suppose it is
    true before some iteration for a set $S$. Let $T$ and $X$ be as in
    the claim. If $i \in M$ is chosen in Step~\ref{step:makeS} then,
    from the description of the procedure $\Call{IsAddableQ}{\ldots}$
    in Algorithm~\ref{alg:buildlayer}, we can conclude that
    \begin{enumerate}
    \item $i \in  M_\sigma^s \setminus (L_{\leq \ell} \cup I_{\leq \ell} \cup S)$,
    \item
      $|H^b_\sigma(B_{\leq \ell} \cap M_\sigma^b, S \cup \{i\})| =
      |H^b_\sigma(B_{\leq \ell} \cap M_\sigma^b, S)| + 1,$ and
    \item
      $|H^s_\sigma(S \cup \{i\}, T \setminus \{i\})| \; \geq \;
      |H^s_\sigma(S, T)| \; + \; \left(p(\sigma^{-1}(i)) - (\tau^* -1 +
        R - \delta)\right)$.
    \end{enumerate}
    By the induction hypothesis and the first property, $i$ cannot be
    a sink of some flow path in $X$. So, $X$ is a valid flow in
    $H^s_\sigma(S \cup \{i\}, T \setminus \{i\})$. Using the third
    property we therefore conclude that $X$ can be augmented to a
    maximum flow $X'$ in $H^s_\sigma(S \cup \{i\}, T \setminus \{i\})$
    such that
    \[\epsilon |X'| \geq \epsilon |X| +
      \left(p(\sigma^{-1}(i))-(\tau^*-1+R-\delta)\right) \geq \sum_{i'
        \in S \cup \{i\}}
      \left(p(\sigma^{-1}(i'))-(\tau^*-1+R-\delta)\right),\] where the
    second inequality uses the induction hypothesis. In
    Step~\ref{step:smallupdate}, a call to
    $\Call{SmallUpdate}{\ldots}$ is made. In this call, a maximum flow
    in $H^s_\sigma(S \cup \{i\}, T \setminus \{i\})$, say $\bar{X}$,
    is computed at the beginning of the \textbf{while} loop in
    Algorithm~\ref{alg:smallupdate}. By using
    Proposition~\ref{prop:smallupdate}\ref{smallupdate:2}, we can
    conclude that a partial schedule $\sigma'$ and a maximum flow
    $\bar{X}'$ in $H^s_{\sigma'}(S \cup \{i\}, T \setminus \{i\})$ are
    computed at the of the \textbf{while} loop which satisfy the property
    \[\epsilon |\bar{X}'| \geq \sum_{i' \in S \cup \{i\}}
      \left(p({\sigma'}^{-1}(i'))-(\tau^*-1+R-\delta)\right).\]

    Furthermore, Proposition~\ref{prop:smallupdate}\ref{smallupdate:3}
    implies that $f^{\text{sink}} \in M^b_\sigma$ for all
    $f\in \bar{X}'$. This is because the vertex capacities of small
    machine sinks $i'$ is defined to be
    $\tau^*+R-p(\sigma'^{-1}(i'))$ in
    Definition~\ref{def:smallflownetwork},
    $p(\sigma'^{-1}(i')) > \tau^* + R - \epsilon$, and flow paths
    carry flows of value $\epsilon$.
  \end{proof}

  Let $L_0,\dots,L_\ell$ denote the set of layers at the beginning of
  the current iteration. Fix some $1 \leq i \leq \ell$. By
  Lemma~\ref{lem:noncollapsibility},
  \[|\{ i' \in A_i \; | \; p(\sigma^{-1}(i')) \leq \tau^* + R - 1\}| <
    \mu_2|A_i|.\]

  Now, consider the iteration (some previous one) in which $L_i$ was
  constructed and let $\sigma^{(b)}$ be the partial schedule at the
  end of Step~\ref{step:assignS} in Algorithm~\ref{alg:buildlayer}
  during the corresponding call to $\Call{BuildLayer}{\ldots}$. Using
  Claim~\ref{claim:inv1}, after the assignments in
  Steps~\ref{step:assignS} and~\ref{step:makeB}, $X$ is a maximum flow
  in $H^s_{\sigma^{(b)}}(A_i, B_i)$ such that
  \[\epsilon |X| \geq \sum_{i' \in A_i}
    \left(p({\sigma^{(b)}}^{-1}(i))-(\tau^*-1+R-\delta)\right) >
    \delta (1-\mu_2)|A_i| - 2\mu_2|A_i|,\] where we use
  Lemma~\ref{lem:noncollapsibility} in the final step along with the
  bound $(\tau^*-1+R-\delta) \leq 2$. Now consider a sink
  $f^{\text{sink}} \in M^b_{\sigma^{(b)}}$ used by some flow path
  $f \in X$. By Proposition~\ref{prop:smallupdate}\ref{smallupdate:3},
  $p({\sigma^{(b)}}^{-1}(f \cap T)) > \tau^* + R
  -\epsilon$. Definition~\ref{def:smallflownetwork} states that the
  vertex capacity
  $c(f^{\text{sink}}) = 1 + \tau^* + R - p(\sigma'^{-1}(f \cap
  T))-\epsilon$ since $f^{\text{sink}} \in M^b_{\sigma^{(b)}}$ from
  Claim~\ref{claim:inv1}. Thus, $c(f^{\text{sink}}) < 1$. This proves
  that at the iteration in which $L_i$ was constructed, by flow conservation,
  $|B_i \cap M^b_{\sigma^{(b)}}| > (\delta(1-\mu_2) - 2\mu_2)|A_i|$.

  In the intervening iterations, the variable $r$ in
  Step~\ref{step:chooser} of Algorithm~\ref{alg:localsearch} might
  have been chosen to be $i+1$, and, therefore,
  $|B_i \cap M^b_\sigma|$ may have reduced in
  Step~\ref{step:bigupdate} and
  $|\{ i' \in A_i \; | \; p(\sigma^{-1}(i')) > \tau^* + R - 1\}|$ may
  have reduced in Step~\ref{step:smallupdate2}. In such an event, all
  the layers following layer $L_i$ would have been discarded in
  Step~\ref{step:discard}. But in our current iteration, the set of
  layers is $L_0,\dots,L_\ell$. So, it only remains to prove that
  $|B_\ell \cap M^b_\sigma| >
  \left(\delta(1-\mu_2)-2\mu_2\right)|A_\ell|$ in the current
  iteration after one or more events where $r$ was chosen to be
  $\ell+1$ in the intervening iterations. The claim is true in this
  case too by arguing as follows. In each intervening iteration, where
  $r$ was chosen to be $\ell+1$, after Step~\ref{step:bigupdate}, the
  partial schedule changes from $\sigma^{(b)}$ to $\sigma^{(a)}$. Due
  to the way sink capacities were set in
  Definition~\ref{def:smallflownetwork}, the new flow network
  $H^s_{\sigma^{(a)}}(A_\ell,B_\ell)$ can be obtained from the old
  flow network $H^s_{\sigma^{(b)}}(A_\ell,B_\ell)$ by increasing the
  capacities of the machines
  $M^b_{\sigma^{(b)}} \setminus M^b_{\sigma^{(a)}}$ (those machines in
  $B_\ell$ from which big jobs were moved in
  Step~\ref{step:bigupdate}) by $\epsilon$. Using the same arguments
  as before after applying Lemma~\ref{lem:noncollapsibility} proves
  the lemma.
\end{proof}



\subsubsection{Maintaining Multiple Sets of Disjoint Paths}\label{sec:maintainingmulti}

We now prove an invariant of Algorithm~\ref{alg:localsearch}
concerning the updates performed in Step~\ref{step:bigupdate} through
the statement
$\sigma \gets \Call{BigUpdate}{\sigma, \{f \in X \; | \; f \cap I_r
  \neq \emptyset\}}$.

\begin{theorem}\label{thm:bigupdate}
  At the beginning of each iteration of the main loop of the
  algorithm, for every $0 \leq i \leq \ell-1$,
  \[|H^b_\sigma(B_{\leq i} \cap M^b_\sigma, A_{i+1} \cup I_{\leq i+1})| \geq
    |A_{i+1}|.\]

  Furthermore, at the beginning of each execution of the
  \textbf{while} loop in Step~\ref{step:collapseq} of
  Algorithm~\ref{alg:localsearch}, for all $0 \leq i \leq \ell-1$,
  $|H^b_\sigma(B_{\leq i} \cap M^b_\sigma, A_{i+1} \cup I_{\leq i+1})|
  \geq |A_{i+1}|.$
\end{theorem}

\begin{proof}
  Consider the first statement. Before the first iteration of the main
  loop, there is nothing to prove. Assume the statement to be true at
  the beginning of iteration $u$ of the main loop for some $u \geq
  1$. We will now show that the statement holds at the end of
  iteration $u$ as well.

  Following Step~\ref{step:finishconstruction} of iteration $u$, the
  newly created layer $L_{\ell+1}$ has the property
  $I_{\ell+1}=\emptyset$ and
  $|H^b_\sigma(B_{\leq \ell} \cap M^b_\sigma, A_{\ell+1})| =
  |A_{\ell+1}|,$ by the construction of set $S$ in
  Step~\ref{step:makeS} of Algorithm~\ref{alg:buildlayer}. For
  $ 0 \leq i \leq \ell - 1$, the statement holds by the induction
  hypothesis and the fact that $\sigma$ was not changed in
  Step~\ref{step:finishconstruction} of
  Algorithm~\ref{alg:localsearch} in a way that affects the graph
  $G^b_\sigma$ (indeed, only small jobs are moved). As in
  Step~\ref{step:initI} of the algorithm, we also update $\ell$ to
  $\ell+1$ in this proof, and so we have at the end of
  Step~\ref{step:initI}, for all $0 \leq i\leq \ell-1$,
  \begin{align}
    |H^b_\sigma(B_{\leq i} \cap M^b_\sigma, A_{i+1} \cup I_{\leq
    i+1})| &\geq |A_{i+1}|.\label{eq:startofinnerwhile}
  \end{align}
  Now iteration $u$ of the main loop could potentially involve one or
  more iterations of the \textbf{while} loop in
  Step~\ref{step:collapseq}. If there are none, we are done using
  \eqref{eq:startofinnerwhile}. The rest of the proof follows from
  Claim~\ref{claim:innerwhile} which completes the induction on $u$,
  and also proves the second statement in Theorem~\ref{thm:bigupdate}.
\end{proof}

\begin{claim}\label{claim:innerwhile}
  At the end of the execution of some iteration of the \textbf{while}
  loop in Step~\ref{step:collapseq} of
  Algorithm~\ref{alg:localsearch}, for all $0 \leq i \leq \ell-1$,
  $|H^b_\sigma(B_{\leq i} \cap M^b_\sigma, A_{i+1} \cup I_{\leq i+1})|
  \geq |A_{i+1}|,$ assuming it holds at the beginning
  of the same iteration of the \textbf{while} loop.
\end{claim}
\begin{proof}
  Assume the statement to be true at the beginning of some iteration
  of the \textbf{while} loop in question as stated in the
  hypothesis. We will now show that the statement holds at the end of
  that iteration as well.

  As in Step~\ref{step:newcollapse}, let
  $I \stackrel{\Delta}{=} \{i \in A_{\ell}\; | \;p(\sigma^{-1}(i))
  \leq \tau^* + R - 1\}$. In Step~\ref{step:canonicald}, we have the
  statement
  \[(I'_0, \dots, I'_{\ell}, X) \gets
    \Call{CanonicalDecomposition}{\sigma, \{L_i\}_{i=0}^{\ell},
      \{I_i\}_{i=0}^{\ell}, I},\] which computes a \emph{specific}
  maximum flow $X$ in the network
  $H^b_\sigma(B_{\leq \ell-1}, I'_{\leq \ell})$ which has the
  properties guaranteed by Proposition~\ref{prop:canonicald}. Recall
  that the sets $I'_i$ are precisely the sinks used by flow paths from
  $X$ with sources in $B_{i-1}$. Additionally, for the purposes of
  this proof, define $X_i$ to be the set of flow paths from $X$ that
  end in sinks $I'_i$, for each $1 \leq i \leq \ell$. Observe that
  $X = X_1 \cup \dots \cup X_\ell$. Let $r$ be the choice made in
  Step~\ref{step:chooser}.

  In the algorithm, in Step~\ref{step:bigupdate}, observe that only
  the flow paths $X_r$ are used to modify the partial schedule from
  $\sigma$ to $\sigma'$ and therefore also change the graph
  $G^b_\sigma$ to $G^b_{\sigma'}$. Step~\ref{step:smallupdate2} does
  not alter the set $A_{r-1}$, even though it alters the partial
  schedule $\sigma'$ by moving small jobs. Since it will not be
  important to consider these updates for the proof of
  Claim~\ref{claim:innerwhile}, which is a statement about the flow
  network for big jobs, we simply denote by $\sigma'$ the partial
  schedule at the end of this step. After this update, all layers
  following layer $L_{r-1}$ are discarded in Step~\ref{step:discard}
  and $\ell$ is updated to $r-1$ in Step~\ref{step:resetell}. So, to
  prove Claim~\ref{claim:innerwhile}, we only need to verify that for
  all $0 \leq i \leq r - 2$,
  \begin{align}
    |H^b_{\sigma'}(B_{\leq i}\cap M^b_{\sigma'}, A_{i+1} \cup I'_{\leq
    i+1})| \geq |A_{i+1}|.\label{eq:toprove}
  \end{align}
  Fix $i$ for the rest of the proof. Since
  $0 \leq i \leq r - 2 \leq \ell-2$, we have, by hypothesis,
  $|H^b_{\sigma}(B_{\leq i}\cap M^b_{\sigma}, A_{i+1} \cup I_{\leq
    i+1})| \geq |A_{i+1}|.$ So let $Y$ be the maximum flow in
  $H^b_{\sigma}(B_{\leq i} \cap M^b_{\sigma}, A_{i+1} \cup I_{\leq
    i+1})$, with at least $|A_{i+1}|$ flow paths. We now interpret the
  flow paths $X$ and $Y$ as vertex disjoint paths in $G^b_\sigma$
  using Proposition~\ref{prop:vertexdisjointpaths}.

  For a collection of vertex disjoint paths $P$ in $G^b_\sigma$ let
  $S_P \subseteq M^b_\sigma$ and $T_P \subseteq M^s_\sigma$ denote the
  set of source and sink vertices respectively used by paths in
  $P$. Now, if it happens that $S_Y \subseteq S_{X_{\leq i+1}}$ then
  we are done. This is because there must be some
  $X' \subseteq X_{\leq i+1}$ such that $X'$ has cardinality at least
  $|S_Y| = |Y| \geq |A_{i+1}|$. Also, $i+1 \leq r-1$ and therefore the
  paths $X_r$ used to update the partial schedule $\sigma$ are
  disjoint from the paths $X'$ and hence, $X'$ continues to be present
  in the new graph $G^b_{\sigma'}$ following the update.

  Our goal is to show this generally in
  Claim~\ref{claim:disjoint}. Note that it immediately implies that
  even after updating the partial schedule to $\sigma'$ we will have a
  collection of at least $|Y| \geq |A_{i+1}|$ many vertex disjoint
  paths connecting the sources of $Y$ to sinks in
  $A_{i+1}\cup I'_{\leq i+1}$ in $G^b_{\sigma'}$. This proves
  \eqref{eq:toprove} and completes the proof of the claim.
\end{proof}

\begin{claim}[Source Alignment Lemma]\label{claim:disjoint}
  There is a collection of vertex disjoint paths $D$ in $G^b_\sigma$
  such that $X_r \subseteq D$, $S_Y \subseteq S_{D \setminus X_r}$,
  and each source in $S_Y$ is connected by $D$ to a sink in
  $A_{i+1} \cup I'_{\leq i+1}$.
\end{claim}

\begin{proof}
  Recall the two sets of disjoint paths
  $X = X_1 \cup \dots \cup X_\ell$ and $Y$ we have defined in the
  graph $G^b_\sigma$. Paths in $X_i$ connect sources in $B_{ i-1} \cap M^b_\sigma$ to
  sinks in $I'_i$, whereas paths in $Y$ connect sources from $B_{\leq
    i} \cap M^b_\sigma$ to sinks in $A_{i+1} \cup I_{\leq i+1}$. We
  now describe a finite procedure which demonstrates the existence of
  paths $D$ as in the claim.
  \begin{center}
    \fbox{\parbox{0.8\textwidth}{
        \begin{algorithmic}
          \State $D \leftarrow X$.
          \While{$\exists s \in S_Y \setminus S_D$}
          \State $v \gets s$.
          \While{$\neg (v \in S_D \vee v \in T_Y) $}
          \If{there is an incoming arc of some path $p \in D$ at $v$}
          \State Set $v$ to be the previous vertex along $p$.
          \ElsIf{there is an outgoing arc of some path $q \in Y$ at $v$}
          \State Set $v$ to be the next vertex along $q$.
          \EndIf
          \EndWhile
          \State Augment $D$ using the set of edges traversed in the
          above loop.
          \EndWhile
          \State \Return $D$.
        \end{algorithmic}}}
  \end{center}

  Before we analyze this procedure, we wish to point out that the
  augmenting step is well-defined. Suppose $D$ and $Y$ are sets of
  disjoint paths before some iteration of the outer \textbf{while}
  loop; this is clearly true before the first iteration. Let
  $s \in S_Y \setminus S_D$ be chosen as the starting vertex for $v$.
  If at no point during the execution of the inner \textbf{while}
  loop, the \textbf{if} condition was satisfied, then it is trivial to
  see that the augmentation step is well-defined since the set of
  edges that are traversed in this case simply corresponds to some
  path $y \in Y$, whose source is $s$, and whose edges do not
  intersect with the edges of paths from $D$. On the other hand,
  consider the first time the \textbf{if} condition is satisfied. From
  that point onwards it is easy to see that the inner loop simply
  traverses the edges of some particular path $x \in D$ in reverse
  order until it reaches the source of $x$, and the inner loop
  terminates. Here, we made use of the induction hypothesis that $D$
  is a collection of vertex disjoint paths. Therefore, we can conclude
  that the total set of edges traversed in this case are composed of
  two disjoint parts: i) edges from the prefix of the path $y \in Y$,
  whose source is $s$, followed by ii) edges from the prefix of the
  path $x \in D$ in reverse order. Furthermore, the unique vertex, say
  $v^*$, at which the two parts have an incoming arc must be a machine
  vertex (since job vertices in $G^b_\sigma$ have at most one incoming
  arc and the two parts are disjoint). Also, $v^* \in M^b_\sigma$
  since $v^* \not \in T_Y$, and the paths $y$ and $x$ must intersect
  at the unique edge $e^*=(v^*, j^*)$ where
  $j^* \in J_b : \sigma(j^*) = v^*$. Thus, deleting the set of edges
  traversed in the second part, and adding the set of edges traversed
  in the first part corresponds to a valid augmentation of $D$.

\begin{figure}
\centering
\includegraphics[width=.75\linewidth]{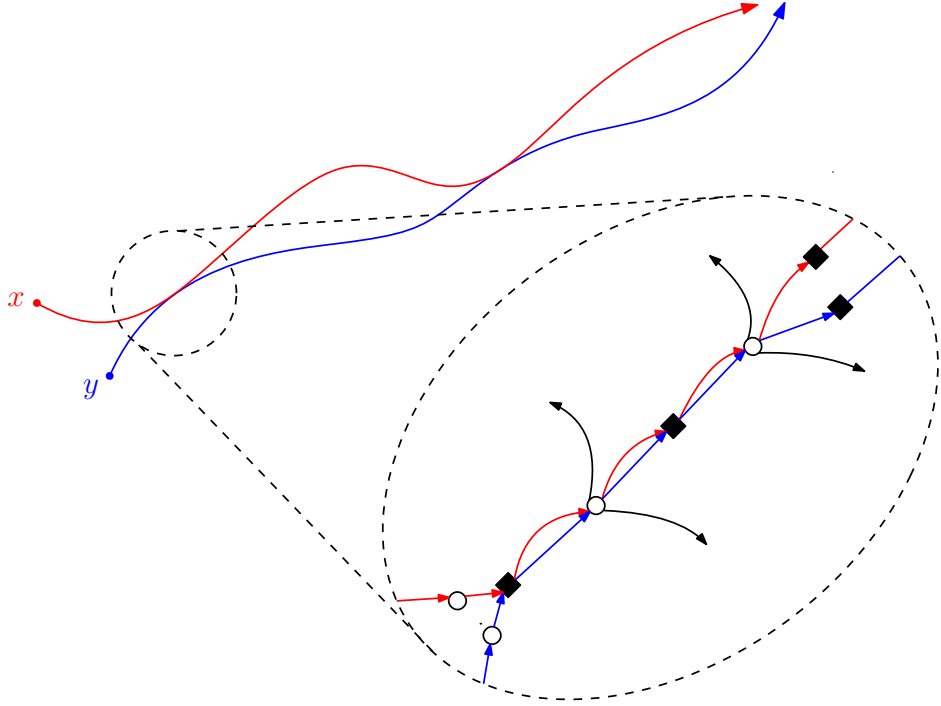}
\caption{The number $\mathcal{N}(x,y)$ of non-contiguous intersections
  between the pair of paths $x \in D$ and $y \in Y$ depicted here is
  $2$. Arcs from $G^b_\sigma$ that are neither present in $D$ nor in
  $Y$ are shown in black.}
\label{fig:sal-1}
\end{figure}

  We now prove that this procedure terminates after finite iterations
  of the outer \textbf{while} loop. We claim that, in each iteration,
  either $|S_Y \setminus S_D|$ decreases, or $|S_Y \setminus S_D|$
  stays the same and the quantity
  \[\mathcal{Q} \stackrel{\Delta}{=} \sum_{y \in Y} \sum_{x \in D}
    \mathcal{N}(x,y)\] decreases, where $\mathcal{N}(x,y)$ is defined
  as the total number of non-contiguous intersections between a pair
  of paths $x$ and $y$ in $G^b_\sigma$ (see
  Figure~\ref{fig:sal-1}). Consider a particular iteration of the
  outer loop. If the \textbf{if} condition is never satisfied during
  the execution of the inner loop, then, by the arguments above, the
  number of disjoint paths in $D$ increases after the augmentation,
  and further the vertex $s$ chosen in the outer loop becomes a new
  source of $D$ after the augmentation. On the other hand, suppose
  that the path chosen for augmentation is composed of two parts
  arising from two paths $y \in Y$ and $x \in D$ as argued
  before. Further, let $s_x$ be the source of $x$, and suppose that
  $s_x \in S_Y$, as otherwise, once again $|S_Y \setminus S_D|$
  decreases after augmenting $D$. Let $y' \in Y$ be the path with
  source $s_x$. After augmenting $D$ it is seen that
  $\sum_{x \in D'} \mathcal{N}(x,y') < \sum_{x \in D}
  \mathcal{N}(x,y'),$ and for all other paths
  $y'' \in Y \setminus \{y'\}$,
  $\sum_{x \in D'} \mathcal{N}(x,y'') \leq \sum_{x \in D}
  \mathcal{N}(x,y'')$, thereby proving that the procedure eventually
  terminates. For an example execution of a single iteration of the
  outer \textbf{while} loop of the procedure, see
  Figure~\ref{fig:sal-2}.

\begin{figure}
  \centering
    \includegraphics[height=.30\paperheight]{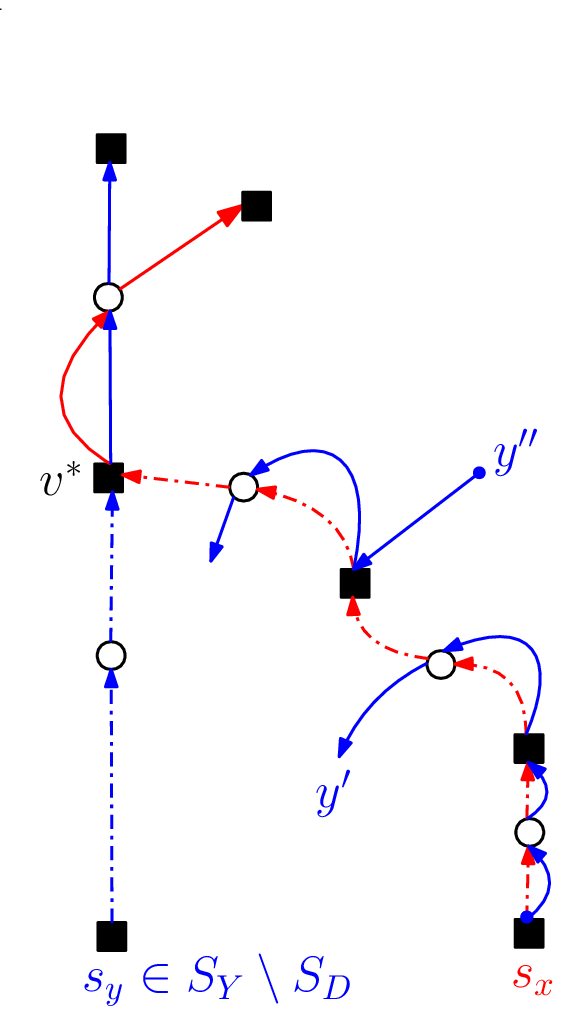}
    \caption{There are three blue paths $y,y',y'' \in Y$ and one red
      path $x \in D$ shown here. Starting from the vertex
      $s_y \in S_Y \setminus S_D$, the outer \textbf{while} loop of
      the procedure defined in Claim~\ref{claim:disjoint} traverses
      precisely the arcs which are dot-dashed and ends in $s_x$, the
      source vertex of $x$, which in this case is the same as the source vertex of $y'$ and therefore $s_x \in S_Y$.}
    \label{fig:sal-2}
\end{figure}

  At the end of the procedure, we have $S_Y \subseteq S_D$. As
  $S_Y \subseteq B_{\leq i}$, by the invariant of the procedure that
  we prove below in Claim~\ref{claim:thed}, the paths in $D$ with
  sources $S_Y$ end in sinks from $A_{i+1} \cup I'_{\leq i+1}$. Also,
  by Claim~\ref{claim:thed}, $X_r \subseteq D$ (because
  $D_3 = X_{\geq i+2} \supseteq X_r$) and
  $S_Y \subseteq S_{D \setminus X_r}$ (because
  $B_{\geq i+1} \supseteq S_{D_3} \supseteq S_{X_r}$), which proves
  the claim.

\begin{claim}[The $D$ Invariant\label{claim:thed}] $D$ is a
  collection of vertex disjoint paths in $G^b_\sigma$ which can be
  partitioned into $D_1 \cup D_2 \cup D_3$ such that:
  \begin{enumerate}[label=(\alph*)]
  \item $S_{D_1} \subseteq B_{\leq i}, \; T_{D_1} = I'_{\leq i+1}$,
  \item $S_{D_2} \subseteq B_{\leq i}, \; T_{D_2} \subseteq A_{i+1}$, and
  \item $D_3 = X_{\geq i+2}$.
  \end{enumerate}
\end{claim}

\begin{proof}
  Before the first iteration of the algorithm, $D$ is initialized to
  $X_{\leq \ell}$ and therefore admits the decomposition into $D_1 =
  X_{\leq i+1}, D_2 = \emptyset, D_3 = X_{\geq i+2}$ which satisfy all
  of the above invariants. Notice here that $i \leq r-2 \leq \ell-2$,
  so that this decomposition is well-defined.

  Assume the $D$ invariant to hold for a collection of disjoint paths
  $D$ at the beginning of some iteration of outer \textbf{while}
  loop. Following the augmentation let $D'$ be the resulting
  collection of disjoint paths. Let
  \begin{align*}
    D'_1 &\stackrel{\Delta}{=} \{p \in D' \; | \; T_{\{p\}} \subseteq I'_{\leq i+1}\},\\
    D'_2 &\stackrel{\Delta}{=} \{p \in D' \; | \; T_{\{p\}} \subseteq A_{i+1}\},\\
    D'_3 &\stackrel{\Delta}{=} \{p \in D' \; | \; T_{\{p\}} \subseteq I'_{\geq i+2}\}.
  \end{align*}

  First we prove that $D'_1, D'_2, D'_3$ defines a valid partition of
  $D'$. As it is clear that the sets are disjoint, we only need to
  prove that $T_{D'} \subseteq A_{i+1} \cup I'_{\leq \ell}$. Recall
  that the augmenting path either ends in a source vertex of $D$ or a
  sink vertex of $Y$. In the first case, no new sinks are introduced,
  i.e., $T_{D'} = T_D \subseteq A_{i+1} \cup I'_{\leq \ell}$. In the
  second case, the augmenting path ends in a sink from
  $S_Y \subseteq A_{i+1} \cup I_{\leq i+1}$. Potentially this could
  introduce a sink from the set
  $I_{\leq i+1} \setminus I'_{\leq i+1}$. But in this case the
  computed canonical decomposition $\{I'_i\}_{i=0}^\ell$ would not be
  maximal since we now have a collection of $|I'_{\leq i+1}| + 1$
  vertex disjoint paths in $D'$ with sources from $B_{\leq i}$ and
  sinks in
  $I'_{\leq i+1} \cup I_{\leq i+1} \subseteq I_{\leq \ell} \cup
  I$. That is, this contradicts the fact that
  $|H^b_\sigma(B_{\leq i}, I'_{\leq i+1})| = |H^b_\sigma(B_{\leq i},
  I_{\leq \ell} \cup I)| $ guaranteed by
  Proposition~\ref{prop:canonicald}\ref{canonicald:2}.

  In the remainder of the proof we show that the defined partition of
  $D'$ satisfies the invariants. Since we do not lose any sinks in the
  augmentation of $D$ to $D'$, by a basic property of flow
  augmentation, it is true that $|D'_1| = |I'_{\leq i+1}|$, and
  therefore $T_{D'_1} = I'_{\leq i+1}$. Next,
  $T_{D'_2} \subseteq A_{i+1}$ and $T_{D'_3} \subseteq I'_{\geq i+2}$ follow
  by definition.

  Since the path used to augment $D$ started from a vertex in
  $S_Y \setminus S_D$ and $S_Y \subseteq B_{\leq i}$ it is clear that
  $S_{D'_1}, S_{D'_2} \subseteq B_{\leq i}$ unless we encountered a
  path from $D_3$ during the augmentation process. However, that would
  lead to a contradiction to the property of the canonical
  decomposition $I'_1 \cup \dots \cup I'_\ell$ that
  $|H^b_\sigma(B_{\leq i}, I'_{\leq i+1})| = |H^b_\sigma(B_{\leq i},
  I'_{\leq \ell})|$ (note here again that $i+1 \leq r-1 \leq \ell-1$)
  by Proposition~\ref{prop:canonicald}\ref{canonicald:3}. Therefore,
  we also have that $S_{D'_1}, S_{D'_2} \subseteq B_{\leq
    i}$. Finally, since we did not encounter any edges of $D_3$ during
  the augmentation process, we not only have that
  $S_{D'_3} \subseteq B_{\geq i+1}$ but that $D_3 = D'_3$.
\end{proof}

\end{proof}

\subsubsection{Proportionally Many Options for Big Jobs}\label{sec:proportionally}

\begin{theorem}\label{thm:newlayer}
  Suppose that $1 \leq \tau^* < 2$. At the beginning of each iteration
  of the main loop of the algorithm, for every
  $0 \leq i \leq \ell-1$,
  \[|A_{i+1}| \geq \mu_1|B_{\leq i}|.\]
  The statement also holds at the beginning of each iteration of the
  \textbf{while} loop of Step~\ref{step:collapseq} of Algorithm~\ref{alg:localsearch}.
\end{theorem}

\begin{proof}
  Let $L_0,\dots,L_\ell$ be the set of layers maintained by the
  algorithm at the beginning of iteration of the main loop. It is
  sufficient to prove that following the construction of layer
  $L_{\ell+1}$ in Step~\ref{step:finishconstruction} of
  Algorithm~\ref{alg:localsearch},
  \[|A_{\ell+1}| \geq \mu_1|B_{\leq \ell}|.\] The rest follows by
  applying Proposition~\ref{prop:basicinvs}\ref{basicinvs:2}. Suppose
  that the set $S$ at the end of the \textbf{while} loop in
  Step~\ref{step:makeS} of Algorithm~\ref{alg:buildlayer} is smaller
  than $\mu_1|B_{\leq \ell}|$. We now describe an assignment of values
  to the variables $(y, z)$ from the dual of $\clp{\tau^*}$ (defined
  in Section~\ref{subsec:clp}) in four parts. Then, we will show that
  the final setting of dual variables $(\bar{y},\bar{z})$ obtained in
  this way satisfies
  $\sum_{j \in J} \bar{z}_j - \sum_{i \in M}\bar{y}_i > 0$, while also
  respecting the constraints of \eqref{eq:clpdual}. It then follows
  that the dual of $\clp{\tau^*}$ is unbounded because for any
  $\lambda > 0$, $(\lambda \bar{y}, \lambda \bar{z})$ is a feasible
  dual solution as well. Therefore, by weak duality, $\clp{\tau^*}$
  must be infeasible, a contradiction. We now proceed to execute this
  strategy.

  \paragraph{Part I: Layers} We set positive dual values to all
  machines that appear in the sets $L_{\leq \ell} \cup I_{\leq \ell}$,
  and the corresponding jobs, as follows:

  \[y^{(1)}_i = \begin{cases}
      \tau^*, \; &i \in L_{\leq \ell} \cup I_{\leq \ell},\\
      0, \; &\text{else}.
    \end{cases} \quad \quad z^{(1)}_j = \begin{cases}
      R-\delta, \; &\exists i \in L_{\leq \ell} \; :
      \; j \in \sigma^{-1}(i) \cap J_b,\\
      \epsilon, \; &\exists i \in L_{\leq \ell} \; :
      \; j \in \sigma^{-1}(i) \cap J_s,\\
      0, \; &\text{else}.
    \end{cases}
  \]
  The objective function of the assignment $(y^{(1)}, z^{(1)})$ can be
  lower bounded as:
  \begin{align*}
    \begin{split}
      \sum_{j \in J}z^{(1)}_j - \sum_{i \in M} y^{(1)}_i &\geq
      (2R-\delta-\epsilon-1)|B_{\leq \ell}| -\tau^*|I_{\leq \ell}|
      -(1-\mu_2)(1-R)|A_{\leq \ell}|-\mu_2\tau^*|A_{\leq \ell}|.
    \end{split}
  \end{align*}
  Let us explain the lower bound. For each machine
  $i \in B_{\leq \ell}$, $p(\sigma^{-1}(i)) > \tau^* + R - \epsilon$
  from Proposition~\ref{prop:basicinvs}\ref{basicinvs:3}. This allows
  us to derive
  $\sum_{j \in \sigma^{-1}(i)} z^{(1)}_j-y^{(1)}_i > \tau^* + R
  -\epsilon - (1-(R-\delta))-\tau^* = 2R -\delta - \epsilon - 1$. For
  the machines $i \in I_{\leq \ell}$, we have the trivial lower bound
  $\sum_{j \in \sigma^{-1}(i)} z^{(1)}_j-y^{(1)}_i \geq
  -\tau^*$. Next, for each machine $i \in A_{\leq \ell}$ such that
  $p(\sigma^{-1}(i)) > \tau^* - 1 + R$, we have
  $\sum_{j \in \sigma^{-1}(i)} z^{(1)}_j-y^{(1)}_i \geq \tau^* - 1 + R
  -\tau^* = -1 + R$, whereas for the rest of the machines in
  $A_{\leq \ell}$, we use the trivial lower bound $-\tau^*$. Thus,
  using Lemmas~\ref{lem:noncollapsibility} and~\ref{lem:smallI}, we
  have
  \begin{align}
    \sum_{j \in J}z^{(1)}_j - \sum_{i \in M} y^{(1)}_i
    &\geq \left(2R-\delta-\epsilon-1-\tau^*\mu_1\mu_2\right)|B_{\leq
      \ell}|-(1-\mu_2)(1-R)|A_{\leq \ell}|-\mu_2\tau^*|A_{\leq
      \ell}|.\label{eq:part1}
  \end{align}

  At this point we have assigned a positive $z^{(1)}_j$ value to big
  jobs $j$ assigned by $\sigma$ to machines in $L_{\leq
    \ell}$. However there could potentially be machines
  $i \in M \setminus (L_{\leq \ell} \cup I_{\leq \ell})$ such that
  $i \in \Gamma(j)$ as well. Therefore, the current assignment
  $(y^{(1)}, z^{(1)})$ does not necessarily constitute a dual feasible
  solution since it might violate the inequality
  $y_i \geq \sum_{j \in C}z_j,$ for a configuration
  $C = \{j\} \in \conf{i, \tau^*}$ consisting of a single big job. We now
  fix this in the next part. For convenience, let
  $M^{(1)} \stackrel{\Delta}{=} L_{\leq \ell} \cup I_{\leq \ell}$.

\begin{figure}
  \centering
    \includegraphics[width=.7\linewidth]{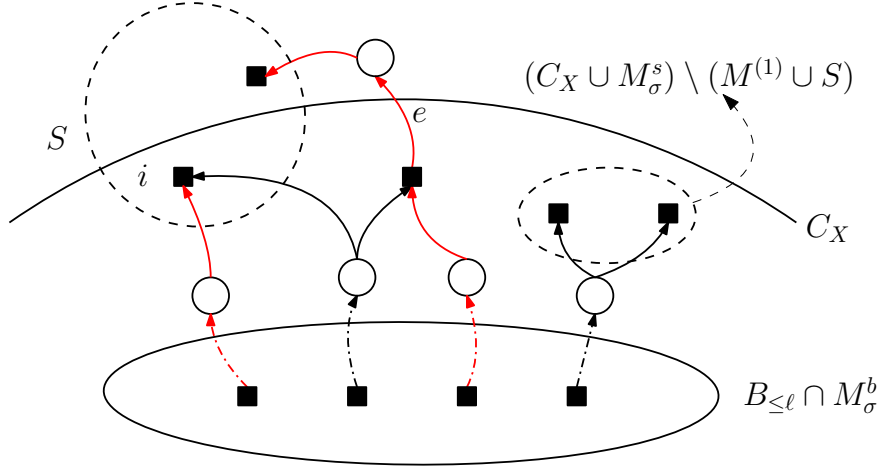}
    \caption{The cut $C_X$ arising from the flow $X$,
      whose paths are colored red, in the flow network
      $H^b_\sigma(B_{\leq \ell} \cap M^b_\sigma, S)$ is shown here. All machines
      except those inside dotted circles are big machines. Every big
      machine is matched to a distinct big job except for the machine
      with the outgoing machine-job arc $e$, which crosses $C_X$. The
      matching edges are dot-dashed and $i$ has unit capacity.}
    \label{fig:apm}
\end{figure}

  \paragraph{Part II: Approximate Matchings} Consider the flow network
  of big jobs $H^b_\sigma(B_{\leq \ell} \cap M^b_\sigma, S)$ that was
  used to construct the set $S$. By the construction of the set $S$ in
  the algorithm, there is a flow $X$ in this network of value
  $|S|$. This flow naturally defines a minimum capacity cut: the cut
  $C_X$ is defined as the set of reachable jobs and machines from
  $B_{\leq \ell} \cap M^b_\sigma$ in the \emph{residual flow network}
  corresponding to $X$ in
  $H^b_\sigma(B_{\leq \ell} \cap M^b_\sigma, S)$. Let
  $M^{(2)} \stackrel{\Delta}{=} (C_X \cap M^b_\sigma) \setminus
  M^{(1)}$. We extend the assignment $(y^{(1)}, z^{(1)})$ described in
  the first part in the following way.
  \[y^{(2)}_i = \begin{cases}
      R-\delta, \; &i \in M^{(2)},\\
      0, \; &\text{else}.
    \end{cases} \quad \quad z^{(2)}_j = \begin{cases}
      R-\delta, \; &\exists i \in M^{(2)} \; : \; j \in \sigma^{-1}(i) \cap J_b,\\
      0, \; &\text{else}.
    \end{cases}
  \]
  The capacity of the cut $C_X$ is $|S|$ by the max-flow min-cut
  theorem. This in particular implies that no job-machine arcs can
  cross this cut as such arcs have infinite capacity. In other words,
  the only arcs crossing $C_X$ are machine-job arcs and
  machine-supersink arcs where the machine arises from $S$ (recall
  that sink vertices have vertex capacity $1$ in
  $H^b_\sigma(B_{\leq \ell} \cap M^b_\sigma, S)$ according to
  Definition~\ref{def:bigflownetwork}).
  
  For every big machine $i$ that is present in $C_X$, the
  corresponding big job assigned to $i$ by $\sigma$ is also present in
  $C_X$ with the exception of at most $|S|$ big jobs as shown in Figure~\ref{fig:apm} (corresponding to
  the machine-job arcs that cross the cut $C_X$). Therefore, the total
  loss incurred in this step is at most $|S|$. In other words,
  \begin{align}
    \sum_{j \in J} z^{(2)}_j - \sum_{i \in M} y^{(2)}_i \geq -(R-\delta)|S|.\label{eq:part2}
  \end{align}

  \paragraph{Part III: Growing Mincuts} In this part and the next, we
  assign positive dual values to machines in $S$, machines in
  $\left(C_X \cap M^s_\sigma\right) \setminus M^{(1)}$, and some other
  machines, to complete the description of our dual assignment. To
  make such an assignment, we will use the algorithm described in Figure~\ref{fig:grmprocedure} in the
  analysis.

  \begin{figure}
    \centering
    \fbox{\parbox{0.9\textwidth}{
        \begin{algorithmic}
          \State $\rho \gets \sigma$.
          \State $U \gets S$.
          \State $V \gets M \setminus (M^{(1)} \cup U)$.
          \State $Y \gets$ Maximum flow in $H^s_\sigma(U, V)$.
          \State $C_Y \gets$ Mincut corresponding to $Y$ in $H^s_\sigma(U,V)$.
          \While{$\exists \; i \in (C_X \cap M^s_\rho) \setminus (M^{(1)}
            \cup C_Y)$}
          \State Augment $Y$ to a maximum flow in $H^s_\rho(U\cup\{i\}, V\setminus\{i\})$.
          \Comment This is well-defined
          \State $C_Y \gets$ Mincut corresponding to $Y$ in $H^s_\rho(U\cup\{i\}, V\setminus\{i\})$.
          \For{$f \in Y : \text{the sink of $f$ belongs to $((C_X \cap M^s_\rho) \setminus (M^{(1)}     \cup C_Y))$}$}
          \State Update $\rho$ by using the flow path $f$.
          \State $Y \gets Y \setminus \{f\}$.
          \EndFor
          \Comment $C_Y$ is still the mincut corresponding to $Y$
          \State $U \gets U \cup \{i\}$.
          \State $V \gets V \setminus \{i\}$.
          \EndWhile
          \State \Return $C_Y$
        \end{algorithmic}}}
    \caption{Mincut Growing Procedure}
    \label{fig:grmprocedure}
  \end{figure}
  The properties of the above procedure that we require in the proof
  are encapsulated in the following claim which we will prove inductively.

  \begin{claim}\label{claim:heart}
    The \textbf{while} loop of the above procedure maintains the
    following invariants.
    \begin{enumerate}[label=(\alph*)]
    \item\label{heart:1} $\rho$ is a partial schedule.
    \item\label{heart:2} $Y$ is a maximum flow in $H^s_\rho(U, V)$ and $C_Y$ is the corresponding mincut.
    \item\label{heart:3} The value of maximum flow $Y$ can be upper bounded as \[|H^s_\rho(U, V)| \leq (\tau^*+R)|S| + \sum_{i \in U \setminus S}(p(\rho^{-1}(i))-(\tau^*-1+R-\delta)).\vspace{-1em}\]
    \item\label{heart:4} There is no flow path $f \in Y$ that ends in a sink belonging to $(C_X \cap M^s_\rho) \setminus (M^{(1)} \cup C_Y)$.
    \item\label{heart:5} For each $i \in (C_X \cap M^s_\rho) \setminus (M^{(1)} \cup C_Y)$,
      \[|H^s_\rho(U \cup \{i\}, V \setminus \{i\})| < |H^s_\rho(U, V)| + (p(\rho^{-1}(i)) - (\tau^* - 1 + R - \delta)).\]
    \end{enumerate}

  \end{claim}
  \begin{proof}
    Before the first iteration of the \textbf{while} loop,
    Claim~\ref{claim:heart}\ref{heart:3} is satisfied because $U = S$
    and $\rho = \sigma$ is a partial schedule;
    Claim~\ref{claim:heart}\ref{heart:4} is satisfied because $Y$ is a
    flow that has the properties guaranteed by Claim~\ref{claim:inv1};
    Claim~\ref{claim:heart}\ref{heart:5} follows from the fact that
    the \textbf{while} loop in Step~\ref{step:makeS} of
    Algorithm~\ref{alg:buildlayer} was exited. The last claim needs
    some more explanation. Note that every
    $i \in (C_X \cap M^s_\sigma) \setminus (M^{(1)} \cup C_Y)$
    satisfies the following two properties:
    \begin{itemize}
    \item $i \in  M_\sigma^s \setminus (L_{\leq \ell} \cup I_{\leq \ell} \cup S)$, and
    \item
      $|H^b_\sigma(B_{\leq \ell} \cap M_\sigma^b, S \cup \{i\})| =
      |H^b_\sigma(B_{\leq \ell} \cap M_\sigma^b, S)| + 1.$
    \end{itemize}
    The second property follows from
    Proposition~\ref{prop:flows}\ref{flows:1} because $i \in
    C_X$. Therefore, for all such $i$, it must be the case that
    $H^s_\sigma(S \cup \{i\}, T \setminus \{i\}) < H^s_\sigma(S, T) \;
    + \; \left(p(\sigma^{-1}(i)) - (\tau^* -1 + R - \delta)\right)$.

    Suppose that the statement is true until the beginning of some
    iteration of the \textbf{while} loop. Let $Y_s$ be the maximum
    flow in $H^s_{\rho_s}(U_s, V_s)$ and $C_{Y_s}$ be the
    corresponding minimum cut maintained by this procedure. We now
    show it holds at the end of that iteration as well. Let the
    machine chosen in this iteration be
    $i \in (C_X \cap M^s_{\rho_s}) \setminus (M^{(1)} \cup
    C_{Y_s})$. The augmentation step is well defined because $Y_s$ is
    a flow in $H^s_{\rho_s}(U_s, V_s)$ that does not use any flow path
    with a sink belonging to
    $(C_X \cap M^s_{\rho_s}) \setminus (M^{(1)} \cup C_{Y_s})$ as
    guaranteed by Claim~\ref{claim:heart}\ref{heart:4}. Therefore,
    $Y_s$ is also a feasible flow in
    $H^s_{\rho_s}(U_s \cup \{i\}, V_s \setminus \{i\})$, which can be
    augmented to a maximum flow in that network, say $Y'$. Let
    $C_{Y'}$ be the corresponding mincut that is computed in the
    procedure. We remark that $C_{Y_s} \subset C_{Y'}$ where the
    inclusion is strict, because $i \in C_{Y'} \setminus C_{Y_s}$,
    which follows from Proposition~\ref{prop:flows}\ref{flows:2}. We
    use this fact later.

\begin{figure}
  \centering
  \includegraphics[width=.7\linewidth]{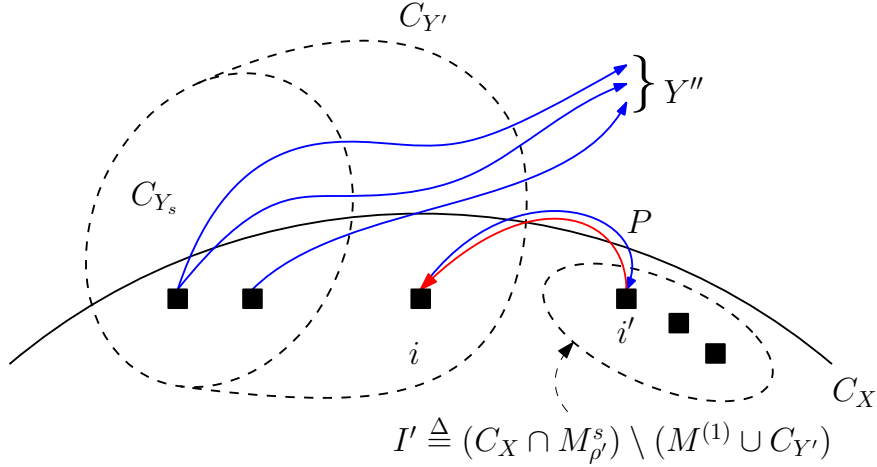}
    \caption{The induction step in the proof of
      Claim~\ref{claim:heart}. For an
      $i \in (C_X \cap M^s_{\rho^s}) \setminus (M^{(1)} \cup
      C_{Y_s})$, $Y'$ is the maximum flow in
      $H^s_{\rho_s}(U_s \cup \{i\}, V_s \setminus \{i\})$ shown in
      blue. It is partitioned as $Y'' \cup P$. The direction of the
      arcs in $P$ is reversed, as shown in red, during the
      \textbf{for} loop from the procedure defined in Part III. The
      mincuts corresponding to the flows $Y'$ and $Y_s$ obey the inclusion
      $C_{Y'} \supset C_{Y_s}$.}
    \label{fig:grm}
  \end{figure}

    The execution now enters a \textbf{for} loop that modifies the
    flow $Y$ and the partial schedule $\rho$ maintained by the
    procedure, which currently assume the values $Y'$ and $\rho_s$
    respectively. Let $Y'' \subseteq Y'$ and $\rho'$ be the state of
    $Y$ and $\rho$ at the \emph{end} of the \textbf{for} loop
    respectively, so that $P \stackrel{\Delta}{=} Y' \setminus Y''$ is
    precisely the set of flow paths used to update the partial
    schedule $\rho$ maintained by the procedure. It is seen that the
    $G^s_{\rho'}$ is obtained from the graph $G^s_{\rho_s}$ by
    reversing the directions of the arcs contained in the paths of
    $P$. For an illustration, see Figure~\ref{fig:grm}.


    As only flow paths $P \subseteq Y'$ ending in small machine sinks
    were used to update $\rho$, $\rho'$ is still a partial schedule
    (recall that $Y'$ is a flow in the network
    $H^s_{\rho_s}(U_s \cup \{i\}, V \setminus \{i\})$), which proves
    Claim~\ref{claim:heart}\ref{heart:1}. We also claim that
    Claim~\ref{claim:heart}\ref{heart:2} holds.
    \begin{claim}\label{claim:heartinternal}
      $Y''$ is a maximum flow in
      $H^s_{\rho'}(U_s \cup \{i\}, V_s \setminus \{i\})$ and $C_{Y'}$
      is the corresponding minimum cut.
    \end{claim}
    \begin{proof}
      The first part is true, because, after every single iteration of
      the \textbf{for} loop, the value of the flow $Y$ decreases by
      $\epsilon$ and so does the capacity of the arcs in
      $H^s_\rho(U \cup \{i\}, V\setminus \{i\})$ crossing the cut
      $C_Y$; since the value of the flow $Y'$ was equal to the
      capacity of the minimum cut $C_{Y'}$ \emph{before} the
      \textbf{for} loop, by the max-flow min-cut theorem, the claim
      follows (we apply the converse of the max-flow min-cut theorem
      at the end).

      The second part is true as well, because the set of vertices
      reachable from $U_s \cup \{i\}$ in the reduced flow network
      corresponding to the flow $Y'$ in
      $H^s_{\rho_s}(U_s \cup \{i\}, V_s \setminus \{i\})$ (defined to be
      $C_{Y'}$) is the \emph{same} as the set of vertices reachable
      from $U_s \cup \{i\}$ in the reduced flow network corresponding
      to the flow $Y''$ in
      $H^s_{\rho'}(U_s \cup \{i\}, V_s \setminus \{i\})$.
    \end{proof}

    Next, using Claims~\ref{claim:heart}\ref{heart:5} and~\ref{claim:heart}\ref{heart:3}, we see that
    \begin{align*}
    |H^s_{\rho_s}(U_s \cup \{i\}, V_s \setminus \{i\})|
      &<|H^s_{\rho_s}(U_s, V_s)| + \underbrace{(p(\rho_s^{-1}(i))-(\tau^*-1+R-\delta))}_{(*)}\\
      &\leq (\tau^*+R)|S| + \sum_{i' \in U_s \setminus S}(p(\rho_s^{-1}(i'))-(\tau^*-1+R-\delta)) + (*)\\
      &= (\tau^*+R)|S| + \sum_{i' \in (U_s \cup \{i\} )\setminus S}(p(\rho_s^{-1}(i'))-(\tau^*-1+R-\delta)).
    \end{align*}
    In the final equality, note that $i \not \in S$ because
    $i \not \in C_{Y'}$ and $C_{Y'} \supseteq U_s \supseteq S$ (using
    the fact mentioned in our earlier remark). In each iteration of
    the \textbf{for} loop, we saw that the value of the quantities
    $|H^s_{\rho}(U \cup \{i\}, V \setminus \{i\})|$ and
    $p(\rho^{-1}(i))$ reduces exactly by $\epsilon$ so that at the end
    of the \textbf{for} loop, we have
    \[ |H^s_{\rho'}(U_s \cup \{i\}, V_s \setminus \{i\})| \leq (\tau^*
      + R)|S| + \sum_{i' \in (U_s \cup \{i\} )\setminus
        S}(p(\rho'^{-1}(i'))-(\tau^*-1+R-\delta)),\] which proves
    Claim~\ref{claim:heart}\ref{heart:3}.

    At the end of the \textbf{for} loop,
    Claim~\ref{claim:heart}\ref{heart:4} is true as well because for
    every $f \in Y''$ the sink of $f$ does not belong to
    $I' \stackrel{\Delta}{=} (C_X \cap M^s_{\rho'}) \setminus (M^{(1)}
    \cup C_{Y'})$ by the postcondition of the loop, and using
    Claim~\ref{claim:heartinternal}. It only remains to prove
    Claim~\ref{claim:heart}\ref{heart:5}. Suppose towards
    contradiction that there is an $i' \in I'$ such that
    \[|H^s_{\rho'}(U_s \cup \{i,i'\}, V_s \setminus \{i,i'\})| \geq
      |H^s_{\rho'}(U_s \cup \{i\}, V_s \setminus \{i\})| +
      (p({\rho'}^{-1}(i')) -(\tau^*-1+R-\delta)).\] Let $F$ be some
    maximum flow in
    $H^s_{\rho'}(U_s \cup \{i,i'\}, V_s \setminus \{i,i'\})$ obtained
    by augmenting $Y''$ (which is well-defined because $Y''$ does not
    have flow paths that use vertices in $I'$ as sinks). Construct a
    new flow network $H'$ from
    $H^s_{\rho_s}(U_s \cup \{i, i'\}, V_s \setminus \{i, i'\})$ by
    adding a copy of the vertex $i'$ and call it $i'_{\text{dummy}}$
    (with identical neighborhood structure and vertex
    capacity). Interpret the flow $P$ in $H'$ so that none of the flow
    paths use $i'$ as as sink (they may use $i'_{\text{dummy}}$
    however). In the \emph{residual flow network} corresponding to
    this flow in $H'$, use the flow paths of $F$ to augment the
    flow. This is well-defined because of the way the graphs
    $G^s_{\rho'}$ and $G^s_{\rho_s}$ are related (recall that, to
    obtain the former from the latter, we just need to reverse the
    directions of the arcs of paths in $P$). It is important to note
    here that the resulting flow only contains paths and no
    cycles. Through this process we obtain a flow of value $|F| + |P|$
    in the network $H'$. By assumption,
    \begin{align*}
      |F| + |P| &\geq |Y''| + \underbrace{(p({\rho'}^{-1}(i')) -(\tau^*-1+R-\delta))}_{(*)} + |P|\\
                &= |Y'| + (*)\\
                &= |Y_s| + f_i + (*).
    \end{align*}
    The first equality follows from the definition of $P$; the second
    equality follows from the fact that $Y_s$ was augmented in the
    flow network $|H^s_{\rho_s}(U_s \cup \{i\}, V_s \setminus \{i\})|$
    to $Y'$ so that $Y'$ has exactly $f_i$ value flow paths with
    sources at $i$, and $|Y_s|$ value flow paths with sources in $U_s$
    (here we make use of the fact that $Y_s$ is a maximum
    flow). Therefore, we have a flow of value at least
    $|Y_s| + f_i + (*)$ in the network $H'$. Since the latter flow was
    constructed by augmenting maximum flows, we can deduce that the it
    is composed of $|Y_s|$ value flow paths originating at $U_s$,
    $f_i$ value flow paths originating at $i$ and the rest originating
    at $i'$. Deleting all flow paths leading to $i'_{\text{dummy}}$,
    we have a resulting flow of value at least
    $|Y_s| + f_i + (*) - l_{i'}$, where $l_{i'}$ is the value of flow
    paths in $P$ that end in $i'$. Owing to the way in which we
    updated $\rho$ in the \textbf{for} loop, we can see that
    $l_{i'} = p(\rho'^{-1}(i')) - p(\rho_s^{-1}(i'))$. Therefore there
    is a flow of value at least
    $|Y_s| + (p({\rho_s}^{-1}(i')) -(\tau^*-1+R-\delta))$ in the
    network $H^s_{\rho_s}(U_s \cup \{i'\}, V_s \setminus \{i'\})$,
    which then implies
    \[|H^s_{\rho_s}(U_s \cup \{i'\}, V_s \setminus \{i'\})| \geq
      |H^s_{\rho_s}(U_s, V_s)| + (p({\rho_s}^{-1}(i'))
      -(\tau^*-1+R-\delta)),\]
    contradicting Claim~\ref{claim:heart}\ref{heart:5}.
  \end{proof}

  Returning to our proof, we now run this procedure with one
  modification: we add $i$ to the set $U$ maintained by the procedure
  only if, in addition to the condition in the \textbf{while} loop,
  the new set $C'_Y$ would have a size at least $|C_Y| + 2$. Let
  $\rho_f$ and $C_{Y,f}$ be the reallocation policy and the cut at the
  end of the exeuction of this modified procedure. We have as a
  postcondition that executing the body of the \textbf{while} loop
  once with an
  $i \in (C_X \cap M^s_{\rho_f}) \setminus (M^{(1)} \cup C_{Y,f})$
  would result only in a set of size $|C_{Y,f}| + 1$ (note that $i$
  would be the new element in that case). Extend the dual
  assignment as follows.
  \[y^{(3)}_i = \begin{cases}
      \tau^*, \; &i \in C_{Y,f},\\
      0, \; &\text{else}.
    \end{cases} \quad \quad z^{(3)}_j = \begin{cases}
      \epsilon, \; &j \in \rho_f^{-1}(i) \cap J_s \; : \; i \in C_{Y,f},\\
      0, \; &\text{else}.
    \end{cases}
  \]

  We now need to bound the total loss incurred in this part of the
  proof. Suppose the \textbf{while} loop in the procedure executes
  $t \geq 0$ times. Let $U_f$, $V_f$, and $Y_f$ be the state of the
  (remaining) variables at the end of the procedure. For convenience
  assume that $U_f = S \cup \{i_1,\dots, i_t\}$, where the numbering
  follows the order in which the machines are added to the variable
  $U$ in the \textbf{while} loop. By
  Claim~\ref{claim:heart}\ref{heart:2}, $Y_f$ is a maximum flow in
  $H^s_{\rho_f}(U_f, V_f)$ and $C_{Y,f}$ is the corresponding
  mincut. Let
  $P \stackrel{\Delta}{=} (C_{Y,f} \setminus U_f) \cap M^s_{\rho_f}$
  and $Q \stackrel{\Delta}{=} C_{Y,f} \cap M^b_{\rho_f}$. Since the
  size of the mincut in the variable $C_Y$ increased by at least $2$
  in each iteration, we have at the end that $|P| + |Q| \geq t$.

  By the max-flow min-cut theorem, the value of the maximum flow
  equals the capacity of the minimum cut, and therefore, by
  Claim~\ref{claim:heart}\ref{heart:3},
  \begin{multline*}
    (\tau^*+R)|S| + \sum_{j=1}^t
    (p(\rho_f^{-1}(i_j))-(\tau^*-1+R-\delta)) \\ > \sum_{i
      \in S\cup\{i_1,\dots,i_t\}\cup P \cup Q} c_i + \sum_{i \in
      P}(\tau^*+R-p(\rho_f^{-1}(i))) + \sum_{i \in
      Q}(\tau^*+1+R-p(\rho_f^{-1}(i))-\epsilon),
  \end{multline*}
  where $c_i$ is the total capacity of machine-job arcs with $i$ as
  one endpoint crossing the minimum cut $C_{Y_f}$. The terms on the left
  together upper bound the value of maximum flow in the final network
  $H^s_{\rho_f}(U_f, V_f)$,
  whereas the terms on the right count the contributions to the
  minimum cut arising from machine-job arcs and machine-sink
  arcs. Splitting the first sum on the right,
  \begin{multline*}
    (\tau^*+R)|S| + \sum_{j=1}^t
    (p(\rho_f^{-1}(i_j))-(\tau^*-1+R-\delta)) \\ > \sum_{i \in S} c_i +
    \sum_{i \in \{i_1,\dots,i_t\}\cup P} c_i + \sum_{i \in Q} c_i +
    \sum_{i \in P}(\tau^*+R-p(\rho_f^{-1}(i))) + \sum_{i \in
      Q}(\tau^*+1+R-p(\rho_f^{-1}(i))-\epsilon).
  \end{multline*}
  After rearranging the terms,
  \begin{multline*}
    (\tau^*+R)|S| -t(\tau^*-1+R-\delta)
    > \sum_{i \in S} c_i -\sum_{i \in \{i_1,\dots,i_t\}\cup P} (p(\rho_f^{-1}(i))-c_i)\\
    -\sum_{i \in Q} (p(\rho_f^{-1}(i))-1-c_i) + (|P|+|Q|)(\tau^*+R)-|Q|\epsilon,
  \end{multline*}
  we derive
  \begin{align}\label{eq:gluing}
    \begin{split}
      -\sum_{i \in S} c_i +\sum_{i \in \{i_1,\dots,i_t\}\cup P}
      (p(\rho_f^{-1}(i))-c_i)&+\sum_{i \in Q} (p(\rho_f^{-1}(i))-1-c_i)
      \\
      &> -(\tau^*+R)|S| +t(\tau^*-1+R-\delta) +
      (|P|+|Q|)(\tau^*+R)-|Q|\epsilon \\
      &\geq  -(\tau^*+R)|S| +t(\tau^*-1+R-\delta) +
      (|P|+|Q|)(\tau^*+R-\epsilon).
    \end{split}
  \end{align}
  
  We demonstrate that the assignment $(y^{(3)},z^{(3)})$ amortizes
  itself locally using~\eqref{eq:gluing}.

  \begin{align*}
      \sum_{j \in J} &
      z^{(3)}_j - \sum_{i \in M} y^{(3)}_i \\
      &= \sum_{i \in S\cup\{i_1,\dots,i_t\}\cup P}
      (p(\rho_f^{-1}(i))-c_i) + \sum_{i
        \in Q} (p(\rho_f^{-1}(i))-1-c_i) - \tau^*(|S| + t + |P|+|Q|)\\
      &\geq \underbrace{-\sum_{i \in S}c_i + \sum_{i \in \{i_1,\dots,i_t\}\cup P}
      (p(\rho_f^{-1}(i))-c_i)+ \sum_{i \in Q}
      (p(\rho_f^{-1}(i))-1-c_i)}_{\eqref{eq:gluing}} -\tau^*(|S| + t + |P|+|Q|)\\
      &>-(\tau^*+R)|S|+t(\tau^*-1+R-\delta)+(|P|+|Q|)(\tau^*+R-\epsilon) -\tau^*(|S| + t + |P|+|Q|)\\
      &=-(\tau^*+R)|S|+t(\tau^*-1+R-\delta) + (|P|+|Q|)(R-\epsilon)-\tau^*(|S| + t)\\
      &=-(2\tau^*+R)|S|+t(R-\delta-1) + (|P|+|Q|)(R-\epsilon)\\
      &\geq-(2\tau^*+R)|S|+t(R-\delta-1) + t(R-\epsilon)\\
      &=-(2\tau^*+R)|S|+t\underbrace{(2R-\delta-\epsilon-1)}_{\text{$\geq 0$
        follows from Claim~\ref{claim:positivity}}}\\
      &\geq-(2\tau^*+R)|S|. \numberthis \label{eq:part3}
  \end{align*}

  \paragraph{Part IV: The Rest} As noted in Part III, we may now have
  machines
  $i \in (C_X \cap M^s_{\rho_f}) \setminus (M^{(1)} \cup C_{Y, f})$
  that increase the size of the set $C_{Y,f}$ described in the
  previous part by one. Let $M^{(4)}$ denote the set of such machines;
  note that they must necessarily be a subset of $M^s_\sigma$ (which
  is the same as $M^s_{\rho_f}$). By the postcondition of the modified
  procedure, we deduce that each machine in $M^{(4)}$ has at least
  $\tau^*-1+R-\delta$ processing time small jobs assigned to it by
  $\rho_f$ such that each of those jobs can be assigned to only
  machines in $C_{Y,f} \cup M^{(1)}$ besides itself. Let
  \[S_i \stackrel{\Delta}{=} \{j \in \rho_f^{-1}(i) \cap J_s \; | \;
  \Gamma(j) \subseteq \{i\} \cup C_{Y,f} \cup M^{(1)}\}.\] We set the
  dual values of these machines as follows.

  \[y^{(4)}_i = \begin{cases}
      \sum_{j \in \rho_f^{-1}(i)} z^{(4)}_j, \; &i \in M^{(4)},\\
      0, \; &\text{else}.
    \end{cases} \quad \quad z^{(4)}_j = \begin{cases}
      \epsilon, \; &j \in S_i \; : \; i \in M^{(4)},\\
      0, \; &\text{else}.
    \end{cases}
  \]

  \paragraph{The Dual Assignment} Before we describe our final dual
  assignment $(\bar{y}, \bar{z})$, let us note that the supports of
  $(y^{(1)},z^{(1)})$, $(y^{(3)},z^{(3)})$ and $(y^{(4)}, z^{(4)})$
  are disjoint by construction. Further, observe that the support of
  $(y^{(2)},z^{(2)})$ may only intersect with the support of
  $(y^{(3)},z^{(3)})$, and is disjoint from the other two. However, we
  can assume without loss of generality that they too are disjoint, as
  machines that receive both positive $y^{(2)}$ and $y^{(3)}$ values
  will only help us in the later arguments. The reasoning is that, for
  a machine $i$ such that $y^{(3)}_i=\tau^*$ and $y^{(2)}_i=R-\delta$,
  we will only consider the contribution of $y^{(3)}_i$ to the final
  assignment $\bar{y}$ in the feasibility whereas we will count both
  contributions towards the objective function i.e., we prove that the
  dual objective function of the final assignment is positive even
  after counting an extra contribution of $R-\delta$ for such
  machines. Note that there can be no jobs in the intersection of the
  supports of the dual assignments from the second and third parts. So
  we assume that the supports of the dual assignments from the four
  parts are disjoint. Set $(\bar{y}, \bar{z})$ to be the union of the
  four assignments in the natural way.

  \paragraph{Feasibility} Our assignment $(\bar{y}, \bar{z})$ to the
  dual variables is such that $\sum_{j \in C} \bar{z}_j \leq \tau^*$
  for every $i \in M, C \in \conf{i,\tau^*}$ because
  $\bar{z}_j \leq p_j$ for every $j \in J$. Therefore, the constraints
  of~\eqref{eq:clpdual} involving machines $i$ for which
  $\bar{y}_i = \tau^*$ are satisfied.

  This leaves us to only consider the machines whose dual values were
  set in Parts II and IV. Let $i \in M^{(2)}$ and
  $C \in \conf{i,\tau^*}$. By
  Proposition~\ref{prop:basicinvs}\ref{basicinvs:4}, the construction
  of the cut $C_{Y,f}$ (note that infinite capacity job-machine arcs
  cannot cross this cut), and the dual setting of $z^{(4)}$ (where we
  assigned positive $z^{(4)}_j$ values only to jobs in $S_i$ for some
  $i \in M^{(4)}$), there can be no $j \in C \cap J_s$ such that
  $\bar{z}_j > 0$. As $\tau^* < 2$, there is at most one big job in a
  configuration. Since it is assigned a dual value of $R-\delta$, all
  constraints involving such machines are satisifed. Now let
  $i \in M^{(4)}$ and $C \in \conf{i,\tau^*}$. Recall that
  $\bar{y}_i > \tau^*-1+R-\delta$. If $C$ contains a big job then
  $\sum_{j \in C} \bar{z}_j \leq R - \delta + \tau^* - 1$. If $C$ does
  not contain big jobs, then,
  \[\sum_{j \in C} \bar{z}_j = \sum_{j \in C : \bar{z}_j > 0}
    \bar{z}_j \leq \sum_{j \in \sigma^{-1}(i)} \bar{z}_j =
    \bar{y}_i.\] The inequality in the middle deserves explanation. This
  follows from the assertion that any job
  $j \in C \cap J_s$ such that $\bar{z}_j > 0$ must be part of $\sigma^{-1}(i)$
  by Proposition~\ref{prop:basicinvs}\ref{basicinvs:4}, the
  construction of the cut $C_{Y,f}$, and the dual setting of
  $z^{(4)}$.

  \paragraph{Positivity} Now that we have described our dual
  assignment, we show
  $\sum_{j \in J} \bar{z}_j - \sum_{i \in M} \bar{y}_i > 0$ by
  counting the contributions to the objective function from the dual
  variable settings in each of the four previous parts.
  
  From \eqref{eq:part1} and \eqref{eq:part2}, the total gain in the
  first and second part is at least
  \[\left(2R-\delta-\epsilon-1-\tau^*\mu_1\mu_2\right)|B_{\leq
      \ell}|-(1-\mu_2)(1-R)|A_{\leq \ell}|-\mu_2\tau^*|A_{\leq \ell}|
    - |S|.\] In the third part, using \eqref{eq:part3}, the total loss
  is at most $(2\tau^*+R)|S|$. In the fourth part there is no net loss
  or gain. So, we can lower bound the objective function value of our
  dual assignment $(\bar{y}, \bar{z})$ as follows, making use of
  Lemma~\ref{lem:bi0} and Lemma~\ref{lem:bi} in the second
  inequality.
  \begin{align*}
    \sum_{j \in J} & \bar{z}_j - \sum_{i \in M} \bar{y}_i \\
    &\geq \left(2R-\delta-\epsilon-1-\tau^*\mu_1\mu_2\right)|B_{\leq
      \ell}|-(1-\mu_2)(1-R)|A_{\leq \ell}|-\mu_2\tau^*|A_{\leq \ell}|
    - |S| - (2\tau^* + R)|S|\\
    &\geq \left(2R-\delta-\epsilon-1-\tau^*\mu_1\mu_2-(1+2\tau^*+R)\mu_1\right)|B_{\leq
      \ell}|-(1-\mu_2)(1-R)|A_{\leq \ell}|-\mu_2\tau^*|A_{\leq \ell}|\\
    &\geq
    \left(2R-\delta-\epsilon-1-\tau^*\mu_1\mu_2-(1+2\tau^*+R)\mu_1\right)\left(\delta(1-\mu_2)-2\mu_2\right)\cdot |A_{\leq \ell}|-\left((1-\mu_2)(1-R)-\mu_2\tau^*\right)|A_{\leq \ell}|\\
    &\geq
    \underbrace{\Big(\left(2R-\delta-\epsilon-1-\tau^*\mu_1\mu_2-(1+2\tau^*+R)\mu_1\right)\left(\delta(1-\mu_2)-2\mu_2\right) -(1-\mu_2)(1-R)-\mu_2\tau^*\Big)}_{(*)}|A_{\leq
      \ell}|.
  \end{align*}

  Subtituting the values of $R, \mu_1, \mu_2, \delta$ as defined in
  the statement of Theorem~\ref{thm:one-epsilon} and
  \eqref{eq:mudelta}, one can verify (see
  Claim~\ref{claim:positivity}) that the bracketed expression $(*)$ is
  strictly positive for every $1 \leq \tau^* < 2$, $0 < \epsilon < 1$
  and $\zeta > 0$.






\end{proof}

\subsubsection{Polynomial Bound on Loop Iterations}\label{sec:polybound}

\begin{corollary}\label{corr:ifr}
  Suppose $1 \leq \tau^* < 2$. In each execution of the \textbf{while} loop in
  Step~\ref{step:collapseq} of Algorithm~\ref{alg:localsearch}, the
  \textbf{if} condition in Step~\ref{step:ifr} is satisfied.
\end{corollary}
\begin{proof}
  Consider the beginning of some iteration of the \textbf{while} loop
  in Step~\ref{step:collapseq} of Algorithm~\ref{alg:localsearch} with
  a state $\mathcal{S}$. By the condition in
  Step~\ref{step:collapseq}, $\ell \geq 1$ and
  $|I| \geq \mu_2|A_{\ell}|$, where
  $I \stackrel{\Delta}{=} \{i \in A_{\ell} \; | \; p(\sigma^{-1}(i))
  \leq \tau^*+R-1\}$. Applying Theorem~\ref{thm:bigupdate},
  $|H^b_\sigma(B_{\leq \ell-1} \cap M^b_\sigma, A_{\ell} \cup I_{\leq
    \ell})| \geq |A_{\ell}|.$ Since $|I| \geq \mu_2|A_{\ell}|$ and
  $I \subseteq A_{\ell}$ this means that
  \[|H^b_\sigma(B_{\leq \ell-1} \cap M^b_\sigma, I \cup I_{\leq
      \ell})| \geq \mu_2|A_{\ell}| \geq \mu_1\mu_2|B_{\leq \ell-1}
    \cap M_\sigma^b|,\] where the second inequality follows from
  Theorem~\ref{thm:newlayer}. By
  Proposition~\ref{prop:canonicald}\ref{canonicald:1}, at least one of
  the sets $I'_i$ computed in Step~\ref{step:canonicald} of
  Algorithm~\ref{alg:localsearch} must be of size at least
  $\mu_1\mu_2|B_{i-1} \cap M_\sigma^b|$ for some $1 \leq i \leq \ell$.
\end{proof}

Given the state $\mathcal{S}$ of the algorithm at some point during
its execution, the \emph{signature of a layer} $L_i$ is defined as
\[s_i \stackrel{\Delta}{=} \Bigl\lfloor
  \log_{\frac{1}{1-\mu_1\mu_2}}\left(\left(\frac{1}{\eta}\right)^i\left|B_i
      \cap M^b_\sigma \right|\right)\Bigr\rfloor + i,\] where
$\eta \stackrel{\Delta}{=} \left(\delta(1-\mu_2)-2\mu_2\right)\mu_1 > 0$ by Claim~\ref{claim:positivity2}. The
\emph{signature vector} corresponding to the given state is then
defined as a vector in the following way:
\[s \stackrel{\Delta}{=} (s_0,\dots,s_\ell,\infty).\]

\begin{lemma}\label{lem:siginvs}
  Suppose $1 \leq \tau^* < 2$. The signature vector satisfies the
  following properties.
  \begin{enumerate}[label=(\alph*)]
  \item\label{siginvs:1} At the beginning of each iteration of the
    main loop of the algorithm, $\ell = O(\log |J_b|)$.
  \item\label{siginvs:2} The coordinates of the signature vector are
    well-defined and increasing at the beginning of each iteration of
    the main loop of the algorithm.
  \end{enumerate}
\end{lemma}
\begin{proof}
  Consider the beginning of some iteration of the main loop of the
  local search algorithm. Let $L_0,\dots,L_\ell$ be the set of layers
  maintained by the algorithm. Let $0 \leq i \leq \ell$. Observe that
  from the moment layer $L_i$ was constructed until now, $A_i$ remains
  unmodified (even though the assignment of
  jobs by $\sigma$ to machines in $A_\ell$ may have changed). This is
  because $A_i$ can be modified only if, in some intervening iteration
  of the main loop, the variable $r$ from Step~\ref{step:chooser} of
  Algorithm~\ref{alg:localsearch} is chosen to be $i'$ for some
  $i' \leq i$. But in that case we discard all the layers
  $L_{i'},\dots,L_\ell$ in Step~\ref{step:discard} and this includes
  layer $L_i$ as well. Therefore, for $0 \leq i \leq \ell-1$,
  \[|B_{i+1} \cap M^b_\sigma|
    \stackrel{\substack{\text{Lem~\ref{lem:bi}}}}{>}
    \left(\delta(1-\mu_2)-2\mu_2\right)|A_{i+1}|
    \stackrel{\text{Thm~\ref{thm:newlayer}}}{\geq}
    \left(\delta(1-\mu_2)-2\mu_2\right) \cdot \mu_1|B_{\leq i}|
    \stackrel{(*)}{\geq} \eta |B_{\leq i} \cap M^b_\sigma|.\] The
  second inequality above uses the fact that the layers
  $L_0,\dots, L_\ell$ were not modified since construction as argued
  previously. As the final term in the chain of inequalities above is
  at least $\eta |B_i \cap M^b_\sigma|$, this
  proves~\ref{siginvs:2}. As the sets
  $B_0,\dots,B_\ell$ are disjoint by construction,
  \[|B_{\leq i + 1} \cap M^b_\sigma| = |B_{i+1} \cap M^b_\sigma| +
    |B_{\leq i} \cap M^b_\sigma| \stackrel{\text{Using $(*)$}}{\geq} 
    (1+\eta)|B_{\leq i} \cap M^b_\sigma|.\]
  As $|B_0 \cap M^b_\sigma| \geq 1$ by Lemma~\ref{lem:bi0}, $\ell$ is
  $O(\log_{1+\eta} |M^b_\sigma|) = O(\log |J_b|)$, which
  proves~\ref{siginvs:1}.
\end{proof}

\begin{lemma}\label{lem:sigcount}
  Suppose $1 \leq \tau^* < 2$. Only $\text{poly}(|J_b|)$ many
  signature vectors are encountered during the execution of
  Algorithm~\ref{alg:localsearch}.
\end{lemma}
\begin{proof}
  By Lemma~\ref{lem:siginvs}\ref{siginvs:1}, and the definition of
  $s_i$, each coordinate is at most $O(\log
  |J_b|)$. Lemma~\ref{lem:siginvs}\ref{siginvs:2} also implies that
  the coordinates of the signature vector are increasing at the
  beginning of the main loop. So every signature vector encountered at
  the beginning of the main loop can be unambiguously described as a
  subset of a set of size $O(\log |J_b|)$.
\end{proof}

\begin{lemma}\label{lem:sigdec}
  Suppose $1 \leq \tau^* < 2$. The signature vector decreases in
  lexicographic value across each iteration of the main loop of the
  local search algorithm.
\end{lemma}
\begin{proof}
  Consider the beginning of some iteration of the main loop of the
  local search algorithm with the state $\mathcal{S}$. So,
  $L_0,\dots, L_\ell$ are the set of layers at the beginning of the
  iteration. During the iteration, a single new layer $L_{\ell+1}$ is
  created in Step~\ref{step:finishconstruction}, and zero or more
  layers from the set $\{L_1,\dots,L_\ell,L_{\ell+1}\}$ are discarded
  in Step~\ref{step:discard}. We consider two cases accordingly.

  \begin{itemize}
  \item \textbf{No layer is discarded.} Therefore, at the end of this
    iteration, we will have layers $L_0,\dots,L_{\ell+1}$ and we can
    apply Lemma~\ref{lem:siginvs}\ref{siginvs:2} at the beginning of
    the next iteration to prove this claim. Note here that we used the
    converse of Corollary~\ref{corr:ifr} to deduce that the
    \textbf{while} loop in Step~\ref{step:collapseq} of
    Algorithm~\ref{alg:localsearch} did not execute since no layer was
    discarded.


  \item \textbf{At least one layer is discarded.}  During each
    iteration of the \textbf{while} loop in Step~\ref{step:collapseq},
    for some $1 \leq r \leq \ell + 1$ as chosen in
    Step~\ref{step:chooser}, the \textbf{if} condition in
    Step~\ref{step:ifr} is satisfied by
    Corollary~\ref{corr:ifr}. Therefore, the size of
    $|B_{r-1} \cap M^b_\sigma|$ reduces to at most
    $(1-\mu_1\mu_2)|B_{r-1} \cap M^b_\sigma|$, and the $(r-1)$-th
    coordinate of the signature vector reduces by at least one,
    whereas the coordinates of the signature vector of the layers
    preceding $r-1$ are unaffected. In other words, the signature
    vector at the beginning of the next iteration of the main loop (if
    any) would be \[s' = (s_0,\dots,s_{r-2},s'_{r-1},\infty),\] where
    $r \leq \ell + 1$ and $s'_{r-1} \leq s_{r-1} - 1$.
  \end{itemize}
\end{proof}

An immediate corollary of Lemma~\ref{lem:sigcount} and
Lemma~\ref{lem:sigdec} is that the local search algorithm described in
Section~\ref{sec:lls} terminates after $\text{poly}(|J_b|)$ iterations
of the main loop under the assumption $1 \leq \tau^* < 2$. Notice,
however, that all statements proved in Section~\ref{sec:polybound}
also hold merely given the \emph{conclusions} of Lemma~\ref{lem:bi0}
and Theorem~\ref{thm:newlayer} without necessarily assuming that
$\tau^* \in [1, 2)$.

\subsection{Proof of the Main Theorem}\label{sec:mainthm}
\begin{proof}[Proof of Theorem~\ref{thm:one-epsilon}]
  Let $\mathcal{I}$ be the given instance of the $(1,\epsilon)$ case
  of \problemmacro{restricted assignment makespan minimization} and
  $\texttt{OPT}$ denote the optimum makespan. Assume for the moment
  that $\tau^*$ is known by solving the Configuration LP.  If
  $\tau^* \geq 2$, then the algorithm of Lenstra, Shmoys and
  Tardos~\cite{lenstra1990approximation} for this problem gives an
  $\texttt{OPT} + p_{\max}$ approximation guarantee, which is of
  course at most $1.5\texttt{OPT}$ in this case.

  Suppose instead that $1 \leq \tau^* < 2$. Start with a partial
  schedule $\sigma$ guaranteed by
  Lemma~\ref{lem:partialschedulesexist}. Let $\sigma$ and
  $j_0 \in J_s \setminus \sigma^{-1}(M)$ denote the input partial
  schedule and small job to Algorithm~\ref{alg:localsearch}. From its
  description it is clear that the partial schedule maintained by it
  is modified either in Step~\ref{step:termination} or within the main
  loop. In the main loop, this occurs in exactly three places:
  Step~\ref{step:smallupdate} of Algorithm~\ref{alg:buildlayer};
  Steps~\ref{step:bigupdate} and \ref{step:smallupdate2} of
  Algorithm~\ref{alg:localsearch}. From Lemma~\ref{lem:sigcount} and
  Lemma~\ref{lem:sigdec}, we deduce that the main loop is exited after
  $\text{poly}(|J_b|)$ iterations. Using
  Proposition~\ref{prop:smallupdate}\ref{smallupdate:1},
  Proposition~\ref{prop:bigupdate}\ref{bigupdate:1}, and
  Step~\ref{step:termination} of the local search algorithm, the
  output partial schedule $\sigma'$ therefore satisfies the property
  \[\sigma'^{-1}(M) = \sigma^{-1}(M) \cup \{j_0\}.\]
  Repeating this algorithm a total of $|J_s|$ times yields a schedule
  of makespan at most $\tau^*+R$ for $\mathcal{I}$ in polynomial time.

  However, it is not necessary to know $\tau^*$ in advance by solving
  the Configuration LP. Suppose that $\tau \in [1, 2)$ is a guess on
  the value of $\tau^*$. Let $\mathcal{A}(\tau)$ denote the following
  algorithm. Run the procedure outlined above after substituting
  $\tau$ in place of $\tau^*$ with two modifications to
  Algorithm~\ref{alg:localsearch}: if $|B_0 \cap M^b_\sigma| = 0$
  after Step~\ref{step:init0}, or if
  $|A_{\ell + 1}| < \mu_1 |B_{\leq \ell}|$ in any execution of
  Step~\ref{step:finishconstruction}, then terminate the procedure
  with an error claiming that the guessed value $\tau < \tau^*$.

  Suppose $A(\tau)$ returns a schedule during a binary search over the
  range $\tau \in [1, 2)$, then it is guaranteed to have makespan at
  most $\tau + R$. Note that the range of possible values for $\tau^*$
  is discrete ($1 + k\epsilon$ or $k\epsilon$ for $k \in
  \mathbb{Z}$). As the running time analysis in
  Section~\ref{sec:polybound} of Algorithm~\ref{alg:localsearch}
  depends only on conclusions of Lemma~\ref{lem:bi0} and
  Theorem~\ref{thm:newlayer}, $\mathcal{A}(\tau)$ is always guaranteed
  to terminate in polynomial time irrespective of whether a schedule
  is returned or not. If $A(\tau)$ does not return a meaningful result
  during the binary search then $\tau^* \geq 2$, and it suffices to
  return the schedule computed by the algorithm of Lenstra et
  al.~\cite{lenstra1990approximation}.
\end{proof}

\subsection{Balancing Against Bipartite Matching}\label{sec:balancing}

The approximation guarantee of the local search algorithm from
Section~\ref{sec:lls} deteriorates with increasing $\epsilon$. There
is however a simple algorithm that performs better for the case of
large $\epsilon$.

\begin{theorem}\label{thm:one-epsilonb}
  Let $0 < \epsilon < 1$. The $(1,\epsilon)$ case of
  \problemmacro{restricted assignment makespan minimization} admits a
  $2-\epsilon$ approximation algorithm.
\end{theorem}
\begin{proof}
  Let $\texttt{OPT}$ denote the makespan of an optimum solution to the
  input instance. Guess $\texttt{OPT}$ through binary
  search. Construct a bipartite graph with
  $\lfloor \texttt{OPT}/\epsilon - 1\rfloor$ small nodes and $1$ big
  node for each machine in the input instance. Each small job is
  connected by an edge to all the nodes of all the machines it can be
  assigned to with a finite processing time. Each big job is connected
  by an edge to all the big nodes of all the machines it can be
  assigned to with a finite processing time. It is easy to see that
  there is a perfect matching in this bipartite graph which
  corresponds to a schedule of makespan at most
  \[\left(\Bigl\lfloor \frac{\texttt{OPT}}{\epsilon} - 1\Bigr\rfloor\right)\epsilon + 1
 \leq  \texttt{OPT} -\epsilon + 1   \leq \texttt{OPT} +
 (1-\epsilon)\texttt{OPT} = (2-\epsilon)\texttt{OPT}.\]
\end{proof}


\begin{proof}[Proof of Theorem~\ref{thm:one-epsilonc}]
  Run the algorithm in Theorem~\ref{thm:one-epsilon} with a parameter
  $\zeta'$ on the input instance to obtain a schedule with makespan at
  most $(1 +R(\epsilon, \zeta'))\texttt{OPT}$. Run the algorithm in
  Theorem~\ref{thm:one-epsilonb} to get a $(2-\epsilon)\texttt{OPT}$
  makespan schedule. The better of the two schedules has an
  approximation guarantee that is no worse than
  \[\min\left\{1 + \frac{1}{2} \left(\epsilon + \sqrt{3-2\epsilon}\right) + \zeta', 2 -
      \epsilon\right\}.\]

  Suppose that $2-\epsilon \geq 17/9 + \zeta$. Then,
  $\epsilon \leq 1/9 -\zeta$. So,
  \[1 + \frac{1}{2} \left(\epsilon + \sqrt{3-2\epsilon}\right) + \zeta'
    \leq 1 + \frac{1}{2} \left(\epsilon + \sqrt{3-2\epsilon}\right)\Big|_{\epsilon = 1/9} + \zeta' = 1
    + \frac{1}{2}\cdot \left(\frac19 + \frac53\right) + \zeta' = \frac{17}{9} + \zeta,\] for $\zeta' = \zeta$.
\end{proof}

\section{Conclusion}\label{sec:conclusion}

In this paper we presented a purely flow based local search algorithm
for the $(1,\epsilon$)-case of \problemmacro{restricted assignment
  makespan minimization}. The guarantees achieved by our algorithm
improve significantly over the previous best one due to Chakrabarty et
al.~\cite{chakrabarty20151}. For an illustration of the approximation
profile for the problem, see Figures~\ref{fig:profile-before}
and~\ref{fig:profile-after}.

We remark that the ideas presented in this paper do \emph{not} crucially
depend on the fact that the instances contain exactly two different
job sizes. Nevertheless, we have chosen to present our results in the
$(1,\epsilon)$-case as there are still certain obstructions which
prevent us from achieving $2-\epsilon_0$ guarantees for the
\emph{restricted} case in general. A second reason is that the
algorithm for the case with two different job sizes admits a clean
description in terms of maximum flows, as we have seen earlier in
Section~\ref{sec:lls}.

\begin{figure}
  \centering
  \includegraphics[width=.67\linewidth]{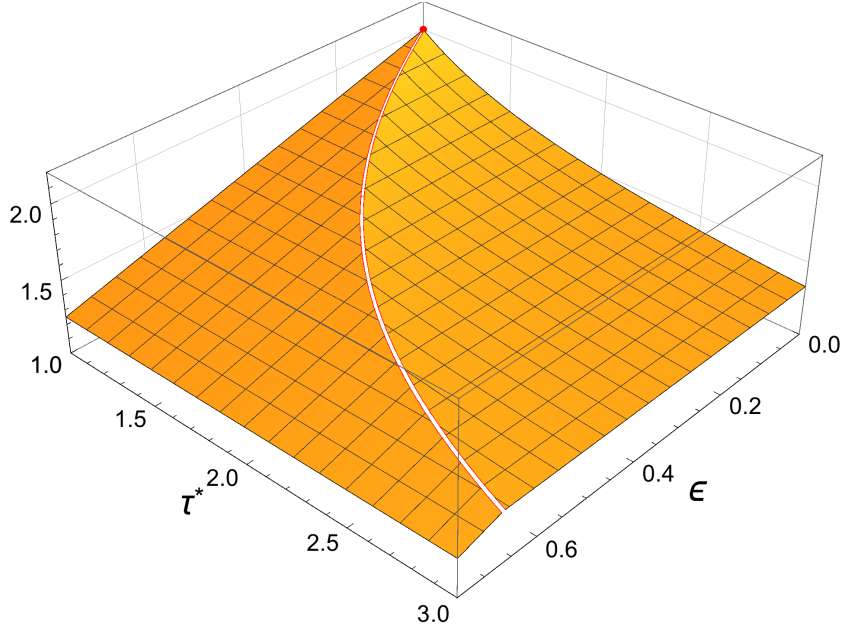}
  \caption{The previous profile of the approximation guarantee as a
    function of $1 \leq \tau^* \leq 3$ and $0 < \epsilon \leq
    3/4$. The two surfaces making up the profile correspond to the
    guarantees $2-\epsilon$ from Theorem~\ref{thm:one-epsilonb} and
    $1+1/\tau^*$ from the algorithm of Lenstra, Shmoys and
    Tardos~\cite{lenstra1990approximation}. The work of Chakrabarty,
    Khanna and Li~\cite{chakrabarty20151} provided a $2-\epsilon_0$ guarantee for some
    positive $\epsilon_0 > 0$, which is indicated as a red dot at the
    apex of the profile.}
    \label{fig:profile-before}
\end{figure}
\begin{figure}
  \centering
  \includegraphics[width=.67\linewidth]{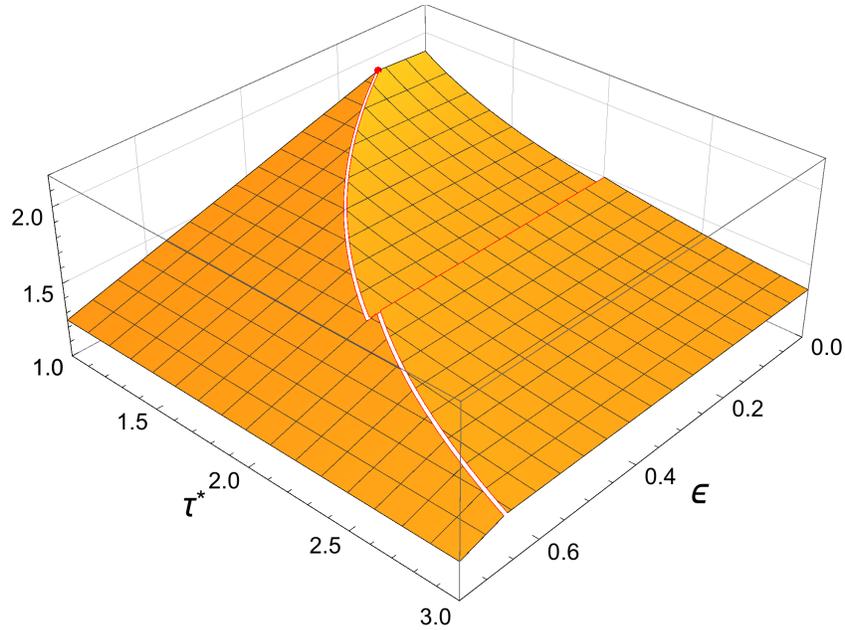}
  \caption{Now, following Theorem~\ref{thm:one-epsilonc}, the worst
    case guarantee is greatest (roughly $1.89$) for instances with
    $\tau^*=1$ and $\epsilon \approx 1/9$ as shown in the figure. The
    third surface arises from the guarantee of
    $1+R(\epsilon, \zeta)/\tau^*$ for $1 \leq \tau^* < 2$ from the
    proof of Theorem~\ref{thm:one-epsilon}.}
    \label{fig:profile-after}
\end{figure}

\bibliographystyle{alpha}
\bibliography{refs}

\begin{thebibliography}{CLRS09}

\bibitem[AFS12]{asadpour2012santa}
Arash Asadpour, Uriel Feige, and Amin Saberi.
\newblock Santa claus meets hypergraph matchings.
\newblock {\em ACM Transactions on Algorithms (TALG)}, 8(3):24, 2012.

\bibitem[AKS15]{DBLP:conf/soda/AnnamalaiKS15}
Chidambaram Annamalai, Christos Kalaitzis, and Ola Svensson.
\newblock Combinatorial algorithm for restricted max-min fair allocation.
\newblock In {\em Proceedings of the Twenty-Sixth Annual {ACM-SIAM} Symposium
  on Discrete Algorithms, {SODA} 2015, San Diego, CA, USA, January 4-6, 2015},
  pages 1357--1372, 2015.

\bibitem[BS06]{bansal2006santa}
Nikhil Bansal and Maxim Sviridenko.
\newblock The santa claus problem.
\newblock In {\em Proceedings of the thirty-eighth annual ACM symposium on
  Theory of computing}, pages 31--40. ACM, 2006.

\bibitem[CKL15]{chakrabarty20151}
Deeparnab Chakrabarty, Sanjeev Khanna, and Shi Li.
\newblock On (1, $\varepsilon$)-restricted assignment makespan minimization.
\newblock In {\em Proceedings of the Twenty-Sixth Annual ACM-SIAM Symposium on
  Discrete Algorithms}, pages 1087--1101. SIAM, 2015.

\bibitem[CLRS09]{DBLP:books/daglib/0023376}
Thomas~H. Cormen, Charles~E. Leiserson, Ronald~L. Rivest, and Clifford Stein.
\newblock {\em Introduction to Algorithms {(3.} ed.)}.
\newblock {MIT} Press, 2009.

\bibitem[EKS08]{ebenlendr2008graph}
Tom{\'a}{\v{s}} Ebenlendr, Marek K{\v{r}}{\'c}al, and Ji{\v{r}}{\'\i} Sgall.
\newblock Graph balancing: a special case of scheduling unrelated parallel
  machines.
\newblock In {\em Proceedings of the nineteenth annual ACM-SIAM symposium on
  Discrete algorithms}, pages 483--490. Society for Industrial and Applied
  Mathematics, 2008.

\bibitem[Fei08]{feige2008allocations}
Uriel Feige.
\newblock On allocations that maximize fairness.
\newblock In {\em Proceedings of the nineteenth annual ACM-SIAM symposium on
  Discrete algorithms}, pages 287--293. Society for Industrial and Applied
  Mathematics, 2008.

\bibitem[HSS11]{haeupler2011new}
Bernhard Haeupler, Barna Saha, and Aravind Srinivasan.
\newblock New constructive aspects of the lovasz local lemma.
\newblock {\em Journal of the ACM (JACM)}, 58(6):28, 2011.

\bibitem[JR16]{jansen2016configuration}
Klaus Jansen and Lars Rohwedder.
\newblock On the configuration-lp of the restricted assignment problem.
\newblock {\em arXiv preprint arXiv:1611.01934}, 2016.

\bibitem[LST87]{lenstra1987approximation}
Jan~Karel Lenstra, David~B Shmoys, and Eva Tardos.
\newblock Approximation algorithms for scheduling unrelated parallel machines.
\newblock In {\em Foundations of Computer Science, 1987., 28th Annual Symposium
  on}, pages 217--224. IEEE, 1987.

\bibitem[LST90]{lenstra1990approximation}
Jan~Karel Lenstra, David~B Shmoys, and {\'E}va Tardos.
\newblock Approximation algorithms for scheduling unrelated parallel machines.
\newblock {\em Mathematical programming}, 46(1-3):259--271, 1990.

\bibitem[PS12]{polacek2012quasi}
Lukas Polacek and Ola Svensson.
\newblock Quasi-polynomial local search for restricted max-min fair allocation.
\newblock In {\em Automata, Languages, and Programming}, pages 726--737.
  Springer, 2012.

\bibitem[Sve12]{svensson2012santa}
Ola Svensson.
\newblock Santa claus schedules jobs on unrelated machines.
\newblock {\em SIAM Journal on Computing}, 41(5):1318--1341, 2012.

\bibitem[WS11]{williamson2011design}
David~P Williamson and David~B Shmoys.
\newblock {\em The design of approximation algorithms}.
\newblock Cambridge University Press, 2011.

\end{thebibliography}

\appendix

\section{Appendix}\label{sec:appendix}

  

\begin{proof}[Proof of Lemma~\ref{lem:partialschedulesexist}]
  Consider the bipartite graph $G=(M \, \cup \, J_b, E)$ where there
  is an edge $\{i,j\} \in E$ if and only if $i \in \Gamma(j)$. A
  perfect matching in $G$ of size $|J_b|$ corresponds to such a
  map. If there is no such perfect matching, by Hall's condition,
  there is a set $S \subseteq J_b$ such that $|N_G(S)| <
  |S|$. Consider the following setting $(y^*, z^*)$ of variables in
  the dual of $\clp{\tau^*}$.
  \begin{align*}
    y^*_i =
    \begin{cases}
      1,  & \text{if $i \in N_G(S)$}, \\
      0, & \text{else}.
    \end{cases}
    \quad \text{and} \quad z^*_j = 
    \begin{cases}
      1,  & \text{if $i \in S$}, \\
      0, & \text{else}.
    \end{cases}
  \end{align*}
  It is now easily verified that $(y^*, z^*)$ is a feasible solution
  to the dual of $\clp{\tau^*}$ defined in~\eqref{eq:clpdual}. We use
  here the fact that configurations $C \in \conf{\tau^*,i}$ for any
  machine $i \in M$ contain at most one big job since $\tau^* < 2$. As
  the objective function value
  $\sum_{j \in J} z^*_j - \sum_{i \in M} y^*_i$ attained by this
  feasible solution $(y^*, z^*)$ is strictly positive, it follows that
  the dual of $\clp{\tau^*}$ is unbounded--for any $\lambda > 0$,
  $(\lambda y^*, \lambda z^*)$ is a feasible dual solution as
  well. Therefore, by \emph{weak duality}, $\clp{\tau^*}$ must be
  infeasible, a contradiction.
\end{proof}

\begin{proof}[Proof of Proposition~\ref{prop:vertexdisjointpaths}]
  By flow decomposition, the maximum flow in $H^b_\sigma(S, T)$ has
  flow paths $p_1,\dots,p_{|H^b_\sigma(S, T)|}$, each of which sends
  one unit of flow from some vertex in $S$ to some vertex in $T$. The
  flow paths may not share a vertex in $T$ as sinks have unit vertex
  capacities in $H^b_\sigma(S, T)$ as defined in
  Definition~\ref{def:bigflownetwork}. Each machine $i \in M$ has at
  most one outgoing edge with unit capacity in $G^b_\sigma$ due to
  Definition~\ref{def:partialschedule}\ref{partialschedule:3} and
  Definition~\ref{def:jobassignmentgraphs}. So the flow paths may also
  not intersect in some vertex in $M \setminus T$ since there is at
  most one outgoing arc with unit capacity. Similarly, they may not
  share a vertex in $J_b$ as there is only one incoming arc of unit
  capacity to a vertex in $J_b$ in $G^b_\sigma$ using
  Definition~\ref{def:partialschedule} and
  Definition~\ref{def:jobassignmentgraphs} .
\end{proof}

\begin{claim}\label{claim:positivity}
  $\forall 1 \leq \tau^* < 2, 0 < \epsilon < 1, \zeta > 0$,
  \[\Big(\left(2R-\delta-\epsilon-1-\tau^*\mu_1\mu_2-(1+2\tau^*+R)\mu_1\right)
      \cdot \left(\delta(1-\mu_2)-2\mu_2\right)-(1-\mu_2)(1-R)-\mu_2\tau^*\Big)
    > 0, \] where
  $\mu_1 = \min\{1,\zeta\}/4, \mu_2 = \min\{\delta,\zeta\}/4, \delta = \left(\sqrt{3-2 \epsilon }-1\right)/2$, and
  $R = \left(\epsilon +\sqrt{3-2 \epsilon }\right)/2 + \zeta.$
\end{claim}
\begin{proof}
  The statement is true if it is true for $\tau^* = 2$. To prove that
  the bracketed expression is positive we substitute the values of
  $\mu_1,\mu_2,\delta$ and $R$ from~\eqref{eq:mudelta} and the
  statement of Theorem~\ref{thm:one-epsilon}, and additionally set
  $\tau^* = 2$ to get the statement
  \begin{multline*}
    -\frac{1}{2} \min \left\{\frac{1}{2} \left(\sqrt{3-2 \epsilon
        }-1\right),\zeta \right\}-\frac{1}{2} \left(2 \zeta +\epsilon
      +\sqrt{3-2 \epsilon }-2\right) \left(\frac{1}{4} \min
      \left\{\frac{1}{2} \left(\sqrt{3-2 \epsilon }-1\right),\zeta
      \right\}-1\right)\\+\frac{1}{16} \left(\frac{1}{4} \left(\sqrt{3-2
          \epsilon }+3\right) \min \left\{\frac{1}{2} \left(\sqrt{3-2
            \epsilon }-1\right),\zeta \right\}-\sqrt{3-2 \epsilon
      }+1\right) \times \\ \left(\min \{1,\zeta\} \left(\min \left\{\frac{1}{2}
          \left(\sqrt{3-2 \epsilon }-1\right),\zeta \right\}+2 \zeta
        +\epsilon +\sqrt{3-2 \epsilon }+10\right)-4 \left(4 \zeta
        +\sqrt{3-2 \epsilon }-1\right)\right) > 0.
  \end{multline*}
  Using a standard computer algebra system for eliminating quantifiers
  over reals, we can verify the truth of the above statement for all
  $0 < \epsilon < 1$ and $\zeta > 0$.
\end{proof}

\begin{claim}\label{claim:positivity2}
  $\forall 0 < \epsilon < 1, \zeta > 0$,
  \[ \delta(1-\mu_2) - 2\mu_2 > 0, \] where
  $\mu_2 = \min\{\delta,\zeta\}/4$, and $\delta = \left(\sqrt{3-2 \epsilon }-1\right)/2$.
\end{claim}
\begin{proof}
  Substituting the values of $\mu_2$ and $\delta$, the statement reads
  \[4 \left(\sqrt{3-2 \epsilon }-1\right) > \left(\sqrt{3-2 \epsilon
      }+3\right) \min \left\{\frac{1}{2} \left(\sqrt{3-2 \epsilon
        }-1\right),\zeta \right\}.\] It suffices to verify that
  statement assuming that the $\min$ term always evaluates to the
  first argument, which then reduces to
  $\epsilon +3 \sqrt{3-2 \epsilon } > 4$. Let $f(\epsilon)$ denote the
  expression on the left. Then,
  $f'(\epsilon) = 1-3/(\sqrt{3-2 \epsilon})$ is negative over
  the range $[0,1]$, $f(0)-4 > 0$ and $f(1)-4=0$.
\end{proof}

\end{document}